\documentclass[12pt]{article}
\usepackage{a4}
\usepackage{amsthm}
\usepackage{amsmath}
\usepackage{amssymb}
\usepackage{amsfonts}
\usepackage{cite}
\usepackage{epsfig}
\usepackage{latexsym}
\usepackage{color}
\usepackage{xspace}
\usepackage{epsfig}
\usepackage[utf8]{inputenc}
\usepackage{cite}
\usepackage{hyperref}
\newtheorem{theorem}{Theorem}
\newtheorem{lemma}[theorem]{Lemma}
\newtheorem{proposition}[theorem]{Proposition}
\newtheorem{corollary}[theorem]{Corollary}

\newcommand{\FF}{\mathbb{F}}
\newcommand{\II}{\mathcal{I}}
\newcommand{\NN}{\mathbb{N}}
\newcommand{\ZZ}{\mathbb{Z}}
\newcommand{\QQ}{\mathbb{Q}}

\newcommand{\hy}{\hbox{-}\nobreak\hskip0pt}
\newcommand{\poly}{\textrm{poly}}

\newcommand{\NPh}{$\mathsf{NP}$\hy{}hard\xspace}
\newcommand{\FPT}{$\mathsf{FPT}$\xspace}
\newcommand{\XP}{$\mathsf{XP}$\xspace}
\newcommand{\W}[1]{$\mathsf{W[#1]}$\xspace}

\newcommand{\cl}{\operatorname{cl}}
\newcommand{\ec}{\operatorname{ec}}
\newcommand{\td}{\operatorname{td}}

\newcommand{\bd}{\operatorname{bd}}
\newcommand{\image}{\operatorname{Im}}
\newcommand{\lin}[1]{{\cal L}\left(#1\right)}

\def\ve#1{\mathchoice{\mbox{\boldmath$\displaystyle\bf#1$}}
{\mbox{\boldmath$\textstyle\bf#1$}}
{\mbox{\boldmath$\scriptstyle\bf#1$}}
{\mbox{\boldmath$\scriptscriptstyle\bf#1$}}}

\newcommand\veb{{\ve b}}

\newcommand\vel{{\ve l}}

\newcommand\veu{{\ve u}}

\newcommand\vex{{\ve x}}

\def\Z{\mathbb{Z}}

\begin{document}
\title{Matrices of optimal tree-depth and a~row-invariant parameterized algorithm for~integer programming\thanks{The first, second and fourth authors were supported by the European Research Council (ERC) under the European Union's Horizon 2020 research and innovation programme (grant agreement No 648509), the second, fourth and fifth by the MUNI Award in Science and Humanities of the Grant Agency of Masaryk University, and
       	the third by Charles University project UNCE/SCI/004 and by 19-27871X of GA \v{C}R. This publication reflects only its authors' view; the ERC Executive Agency is not responsible for any use that may be made of the information it contains. An extended abstract of this work has appeared in the proceedings of ICALP'20.}}
\author{Timothy F.~N. Chan\thanks{School of Mathematical Sciences, Monash University, Melbourne, Australia. E-mail: {\tt timothy.chan@monash.edu}.} \thanks{Mathematics Institute, DIMAP and Department of Computer Science, University of Warwick, UK.}\and
        Jacob W. Cooper\thanks{Faculty of Informatics, Masaryk University, Brno, Czech Republic. E-mails: {\tt jcooper@mail.muni.cz}, {\tt dkral@fi.muni.cz} and {\tt kristyna.pekarkova@mail.muni.cz}.}\and
        Martin Kouteck\'y\thanks{Computer Science Institute, Charles University, Prague, Czech Republic. E-mail: {\tt koutecky@iuuk.mff.cuni.cz}.}\and
\newcounter{kth}
\setcounter{kth}{3}
\newcounter{lth}
\setcounter{lth}{4}
        Daniel Kr\'al'$^{\fnsymbol{lth}\fnsymbol{kth}}$\and
	Krist\'yna Pek\'arkov\'a$^\fnsymbol{lth}$}
        
\date{}

\maketitle

\begin{abstract}
A long line of research on fixed parameter tractability of integer programming culminated with showing that
integer programs with~$n$ variables and a constraint matrix with dual tree-depth~$d$ and largest entry~$\Delta$
are solvable in time~$g(d,\Delta)\text{poly}(n)$ for some function~$g$.
However,
the dual tree-depth of a constraint matrix is not preserved by row operations,
i.e., a given integer program can be equivalent to another with a smaller dual tree-depth, and
thus does not reflect its geometric structure.

We prove that the minimum dual tree-depth of a row-equivalent matrix
is equal to the \emph{branch-depth} of the matroid defined by the columns of the matrix.
We design a fixed parameter algorithm for computing branch-depth of matroids represented over a finite field and
a fixed parameter algorithm for computing a row-equivalent matrix with minimum dual tree-depth.
Finally, we use these results to obtain an algorithm for integer programming
running in time~$g(d^*,\Delta)\text{poly}(n)$
where~$d^*$ is the branch-depth of the constraint matrix;
the branch-depth cannot be replaced by the more permissive notion of branch-width.
\end{abstract}

\section{Introduction}
\label{sec:intro}

Integer programming is a fundamental problem of both theoretical and practical importance.
It is well-known that integer programming in fixed dimension, i.e., with a bounded number of variables,
is polynomially solvable since the work of Kannan and Lenstra~\cite{Kannan:1987,Lenstra:1983} from the 1980s.
Much subsequent research has focused on studying extensions and speed-ups of the algorithm of Kannan and Lenstra. 
However, research on integer programs with many variables has been sparser.
Until relatively recently, the most prominent tractable case has been that of totally unimodular constraint matrices,
i.e., matrices with all subdeterminants equal to~$0$ or~$\pm 1$;
in this case, all vertices of the feasible region are integral and
algorithms for linear programming can be applied.

Besides total unimodularity, several recent results~\cite{HOR,MC,EisHK18,AH,HKW,GOR,GO,DvoEGKO17}
on algorithms for integer programming (IP) exploited various structural properties of the constraint matrix,
yielding efficient algorithms for~$n$-fold IPs, tree-fold IPs, multi-stage stochastic IPs, and
IPs with bounded fracture number and bounded tree-width.
This research culminated in an algorithm by Levin, Onn, and the third author~\cite{KouLO18},
who constructed a fixed parameter algorithm for integer programs
with bounded (primal or dual) tree-depth and bounded coefficients.

These theoretical results well complement a long line of empirical research,
see~\cite{BorFM98,BerCCFLMT15,KhaEE18,FerH98,AykPC04,WeiK71,WanR13,GamL10,VanW10}, demonstrating that
instances of integer programming can be solved efficiently when the constraint matrix is decomposable into blocks.
So, the recent algorithm of Levin, Onn and the third author~\cite{KouLO18}
gives a theoretical explanation for this phenomenon.
In particular, the Dantzig-Wolfe decomposition algorithm is known to work very well
when the constraint matrix has a so-called bordered block-diagonal form with small blocks.
It is usually necessary to describe this form explicitly
but Bergner et al.~\cite{BerCCFLMT15} have shown that
it can be constructed automatically (by permuting the rows of the constraint matrix),
which has had significant performance benefits on important benchmark instances.

However, the bordered block-diagonal form of a matrix and, more generally, the tree-depth of a constraint matrix,
depends on the position of its non-zero entries.
In particular, a matrix with large dual tree-depth, i.e., without any apparent bordered block-diagonal form,
may be row-equivalent to another matrix with small dual tree-depth and thus amenable to efficient algorithms.
We overcome this drawback with tools from matroid theory.
To do so, we consider the \emph{branch-depth} of the matroid defined by the columns of the constraint matrix and
refer to this parameter as the branch-depth of the matrix.
Since this matroid is invariant under row operations,
the branch-depth of a matrix is row-invariant and better captures the true geometry of the instance,
which can be obfuscated by the choice of basis.
Our algorithm thus allows taking the ``automated Dantzig-Wolfe'' approach of Bergner et al.~\cite{BerCCFLMT15} one step further:
it is possible to detect a block structure in a matrix even if it is obscured by row operations, not just by permuting the rows.

Our main results concerning integer programming
can be summarized as follows (we state the results formally in the next subsection).
\begin{itemize}
\item The branch-depth of a matrix~$A$ is equal to
      the minimum dual tree-depth of a matrix row-equivalent to~$A$ (Theorem~\ref{thm:equal}).
\item There exists a fixed parameter algorithm for computing a matrix of minimum dual tree-depth that
      is row-equivalent to the input matrix and
      whose entry complexity stays bounded (Theorem~\ref{thm:algrational}).
\item Integer programming is fixed parameter tractable
      when parameterized by the branch-depth and the entry complexity of the constraint matrix (Corollary~\ref{cor:ip}).
\end{itemize}
Existing hardness results imply that
the parameterization by both branch-depth and entry complexity in Corollary~\ref{cor:ip} is necessary unless \FPT = \W{1},
i.e., it is not sufficient to parameterize only by one of the two parameters.
In particular,
integer programming is \W{1}-hard when parameterized by tree-depth only~\cite{GOR,KnoKM17} and
\NPh for instances with bounded coefficients and dual tree-width (even dual path-width) bounded by two~\cite[Lemma 102]{EisHKKLO19} (also cf.~\cite{KouLO18,GOR}).
The latter also implies that
integer programming is \NPh when the branch-width and the entry complexity of input instances are bounded (also cf.~\cite{FomPRS18}).
On the positive side, Cunningham and Geelen~\cite{CunG07} (also cf.~\cite{MarMH13} for detailed proofs and implementation)
provided a slicewise pseudopolynomial algorithm for non-negative matrices with bounded branch-width,
i.e., the problem belongs to the complexity class \XP for unary encoding of input.
Finally, since the algorithm given in Corollary~\ref{cor:ip} is parameterized
by the branch-depth of the vector matroid formed by the columns of the matrix~$A$,
it is natural to ask whether the tractability also holds in the setting dual to this one,
i.e., when the branch-depth of the vector matroid formed by the \textit{rows} of~$A$ is bounded.
This hope is dismissed in Proposition~\ref{prop:bdp}.

The algorithm from Theorem~\ref{thm:algrational} is based on the following
algorithmic result on matroid branch-depth, which we believe to be of independent interest.
\begin{itemize}
\item There exists a fixed parameter algorithm for computing an optimal branch-depth decomposition of
      a matroid represented over a finite field
      for the parameterization by the branch-depth and the order of the field (Theorem~\ref{thm:algfin}).
\end{itemize}      
To apply this result,
we show that every matroid represented by rational vectors that has bounded branch-depth
is isomorphic to a matroid representable over a finite field (Lemma~\ref{lm:largeq}),
where the order of the field depends on the branch-depth and the entry complexity of the rational vectors.

We would like to point out that the algorithm from Theorem~\ref{thm:algfin} is fully combinatorial,
similarly to the recent algorithm for computing the branch-width of matroids represented over a finite field
by Jeong, Kim and Oum~\cite{JeoKO18},
which extends the classical algorithm for tree-width by Bodlaender and Kloks~\cite{BodK96}, and
unlike the older algorithm by Hlin\v e{}n\'y and Oum~\cite{HliO07,HliO08},
which relies on an exponential upper bound on the size of excluded minors for branch-width by Geelen et al.~\cite{GeeGRW03} and
needs to precompute the list of excluded minors.
While it would likely be possible to follow a similar path in the setting of branch-depth,
we chose the more challenging route of designing a fully combinatorial algorithm,
i.e., one that is based on an explicit dynamic programming procedure.
The benefit of a fully combinatorial approach is that the hidden constants are better,
which is of importance to applications including those in model checking~\cite{GavKO12,Hli03a,Hli03b,Hli06};
in particular, 
Hlin\v en\'y~\cite{Hli03a,Hli03b,Hli06} (in the analogy of Courcelle's result~\cite{Cou90} for graphs) proved that
monadic second order model checking is fixed parameter tractable for matroids with bounded branch-width represented over finite fields.

\subsection{Statement of integer programming results}

To state our integer programming results precisely, we first need to fix some notation.
Vectors throughout our exposition will be written in the bold font.
We consider the general integer programming problem in standard form:
\begin{equation}
\label{IP}
\min\left\{f(\vex) \, \mid A\vex=\veb\,,\ \vel\leq\vex\leq\veu\,,\ \vex\in\ZZ^{n}\right\},
\end{equation}
where~$A\in\ZZ^{m\times n}$ is an integer~$m\times n$ matrix,
$\veb\in\ZZ^m$, $\vel,\veu\in(\ZZ\cup\{\pm\infty\})^n$, and
$f: \ZZ^n \to \ZZ$ is a separable convex function,
i.e., $f(\vex) = \sum_{i=1}^n f_i(x_i)$ where~$f_i: \Z \to \Z$ are convex functions.
In particular, each~$f_i(x_i)$ can be a linear function of~$x_i$.
Integer programming is well-known to be \NPh even
when~$f$ is a constant function and
either the entries of~$A$ are~$0$ and $+1$ (by a reduction from the Vertex Cover problem) or
$m=1$ (by a reduction from the Subset Sum problem).

We next demonstrate the previously mentioned drawback of parameterizing integer programs by tree-depth.
The \emph{tree-depth} of a graph~$G$ is the minimum depth of a rooted forest on the same vertex set such that
the two end-vertices of every edge of~$G$ are in ancestor-descendant relation.
The \emph{dual tree-depth} of a matrix~$A$
is the tree-depth of the graph with vertices corresponding to the rows of~$A$ and with two vertices being adjacent if the corresponding rows have a non-zero entry in the same column.
We define the branch-depth,
which require definitions from matroid theory,
and primal tree-depth of a matrix,
in Section~\ref{sec:notation}.
Consider the following matrices~$A$ and $A'$.
\[
A=\left(\begin{matrix}
  1 & 1 & \cdots & 1 & 1 \\
  2 & 1 & \cdots & 1 & 1 \\
  1 & 2 & \ddots & 1 & 1 \\
  \vdots & \ddots & \ddots & \ddots & 1 \\
  1 & 1 & \ddots & 2 & 1 \\
  1 & 1 & \cdots & 1 & 2
  \end{matrix}\right)
\qquad
A'=\left(\begin{matrix}
   1 & 1 & \cdots & 1 & 1 \\
   1 & 0 & \cdots & 0 & 0 \\
   0 & 1 & \ddots & 0 & 0 \\
   \vdots & \ddots & \ddots & \ddots & 0 \\
   0 & 0 & \ddots & 1 & 0 \\
   0 & 0 & \cdots & 0 & 1
   \end{matrix}\right)
  \]
The dual tree-depth of the matrix~$A$ is equal to its number of rows while the dual tree-depth of~$A'$ is two;
the graphs from the definition of the dual tree-depth are a complete graph and a star, respectively.
We remark that the branch-depth of both matrices is two.
Since the matrices~$A$ and $A'$ are row-equivalent,
the integer programs determined by them ought to be of the same computational difficulty.
More precisely, consider the following matrix~$B$:
\[
B=\left(\begin{matrix}
  1 & 0 & 0 & \cdots & 0 & 0 \\
  -1 & 1 & 0 & \cdots & 0 & 0 \\
  -1 & 0 & 1 & \ddots & 0 & 0 \\
  -1 & \vdots & \ddots & \ddots & \ddots & 0 \\
  -1 & 0 & 0 & \ddots & 1 & 0 \\
  -1 & 0 & 0 & \cdots & 0 & 1
  \end{matrix}\right)
  \enspace.
  \]
Since~$A'=BA$, it is possible to replace an integer program of the form \eqref{IP}
with an integer program with a constraint matrix~$A'=BA$, right hand side~$\veb'=B\veb$, and
bounds~$\vel'=\vel$ and $\veu'=\veu$, and
solve this new instance of IP,
which has dual tree-depth two and the same set of feasible solutions.

In Section~\ref{sec:bdtd}, we first observe that the branch-depth of a matrix~$A$ is at most its dual tree-depth, and
prove that the branch-depth of a matrix~$A$ is actually equal to the minimum dual tree-depth of a matrix~$A'$ that
is row-equivalent to~$A$.
\begin{theorem}
\label{thm:equal}
Let~$A$ be a matrix over a (finite or infinite) field~$\FF$.
The branch-depth of~$A$ is equal to the minimum dual tree-depth of any matrix~$A'$ that is row-equivalent to~$A$.
\end{theorem}
We use the tools developed to prove Theorem~\ref{thm:equal} together with existing results on matroid branch-depth
to obtain an algorithm that given a matrix~$A$ of small branch-depth
outputs a matrix~$B$ that transforms~$A$ to a row-equivalent matrix with small dual tree-depth.
The \emph{entry complexity of a matrix~$A$}, denoted by~$\ec(A)$,
is the maximum length of the binary encoding of an entry~$A_{ij}$ (the length of
binary encoding a rational number~$r=p/q$ with~$p$ and $q$ being coprime is~$\left\lceil\log_2\left(|p|+1\right)\right\rceil+\left\lceil\left(\log_2 |q|+1\right)\right\rceil$).
An algorithm is called \emph{fixed parameter}
if its running time for an instance of size~$n$
is bounded by~$f(k)\poly(n)$
where~$f:\NN\to\NN$ is a computable function and $k$ is a parameter determined by the instance.
\begin{theorem}
\label{thm:alg1}
There exists an algorithm
with running time polynomial in~$\ec(A)$, $n$ and $m$ that for an input~$m\times n$ integer matrix~$A$ and an integer~$d$
either
\begin{itemize}
\item outputs that the branch-depth of~$A$ is larger than~$d$, or
\item outputs an invertible rational matrix~$B\in\QQ^{m \times m}$ such that the dual tree-depth of~$BA$ is at most~$4^d$ and
      the entry complexity of~$BA$ is~$O(d2^{2d}\ec(A))$.
\end{itemize}
\end{theorem}
However, we go further and design a fixed parameter algorithm
for computing the branch-depth of a vector matroid (Theorem~\ref{thm:algrational}) and
use this algorithm to prove the following strengthening of Theorem~\ref{thm:alg1}.
\begin{theorem}
\label{thm:alg2}
There exists a fixed parameter algorithm parameterized by~$d$ and $e$
with running time polynomial in $n$ and $m$ that for an input~$m\times n$ integer matrix~$A$ with entry complexity at most $e$ and an integer~$d$
either
\begin{itemize}
\item outputs that the branch-depth of~$A$ is larger than~$d$, or
\item outputs an invertible rational matrix~$B\in\QQ^{m \times m}$ such that the dual tree-depth of~$BA$
      is equal to the branch-depth of~$A$ and the entry complexity of~$BA$ is~$O(d^22^{2d}\ec(A))$.
\end{itemize}
\end{theorem}
While the algorithm in Theorem~\ref{thm:alg1} runs in polynomial time,
the output matrix can have dual tree-depth (single) exponential in the minimum possible tree-depth;
on the other hand, the running time of the algorithm from Theorem~\ref{thm:alg2}
is triple exponential in~$d$ but it always outputs a row-equivalent matrix with the minimum possible tree-depth.
Also see Table~\ref{tab:complexity} for a comparison of the algorithms from Theorems~\ref{thm:alg1} and \ref{thm:alg2}.

\begin{table}
\begin{center}
\begin{tabular}{|l|cc|}
\hline
Algorithm from & Theorem~\ref{thm:alg1} & Theorem~\ref{thm:alg2} \\
\hline
Hidden constant & none & $2^{\ec(A)2^{2^{2d+2}+O(d)}} $ \\
Tree-depth &	at most~$2^{2d}$	& $d$ \\
Entry complexity & $O\left(d2^{2d}\ec(A)\right)$ & $O\left(d^22^{2d}\ec(A)\right)$ \\
\hline
IP algorithm constant & $2^{O\left(\ec(A)d 2^{\ec(A)d 2^{2d}+2^{2d}+4d}\right)}$
                      & $2^{O\left(\ec(A)d^3 2^{\ec(A)d^2 2^{2d}+3d}\right)}$ \\
\hline
\end{tabular}
\end{center}
\caption{Comparison of algorithms presented in Theorems~\ref{thm:alg1} and \ref{thm:alg2}.
         The first line describes the dependence on the parameter~$d$,
	 the second and the third the tree-depth and the entry complexity of the constraint matrix output by the algorithm, and
	 the last line time needed to solve the instance by the algorithm of Eisenbrand et al.~\cite{EisHKKLO19}.}
\label{tab:complexity}
\end{table}

As explained above, Theorems~\ref{thm:alg1} and \ref{thm:alg2} allow us to perform row operations
to obtain an equivalent integer program with small dual tree-depth from an integer program with small branch-depth.
The function~$g$ depends on which of the theorems is used to find the matrix~$B$ as displayed in Table~\ref{tab:complexity}.
A proof of the corollary using Theorem~\ref{thm:alg1} is presented in Section~\ref{sec:param};
a proof using Theorem~\ref{thm:alg2} is completely the same except that
the obtained matrix~$A'$ has dual tree-depth at most~$d$ and its entry complexity is as given in Table~\ref{tab:complexity}.

\begin{corollary}
\label{cor:ip}
Integer programming is fixed parameter tractable when parameterized by branch-depth and entry complexity,
i.e., an integer program given as in \eqref{IP}
can be solved in time polynomial in~$g(\bd(A),\ec(A))$, $n$, $\ec(\veb)$, $\ec(\vel)$ and $\ec(\veu)$,
where~$\bd(A)$ is the branch-depth of the matrix~$A$ and $g:\NN^2\to\NN$ is a computable function.
\end{corollary}

We remark that the results of~\cite{EisHKKLO19,KouLO18} give a strongly fixed parameter algorithm,
i.e., an algorithm whose number of arithmetic operations does not depend on the size of the numbers involved,
if the objective function~$f$ is a linear function, and
so the algorithm from Corollary~\ref{cor:ip} is strongly polynomial in~$g(\bd(A),\ec(A))$ and $n$
when the objective function is linear.

We note that the dependence of~$g$ on~$\ec(A)$ and $\bd(A)$ is double and triple exponential, respectively,
regardless whether we use Theorem~\ref{thm:alg1} or Theorem~\ref{thm:alg2} to prove Corollary~\ref{cor:ip}.
However, a rather pessimistic estimate on the entry complexity of the integer matrix~$A''$
was used in the proof of Corollary~\ref{cor:ip} and
it is likely that the constraint matrix with a row-equivalent matrix with a significantly smaller dual tree-depth
likely outweighs the increase of the entry complexity
since the parameter dependence in the algorithm of~\cite{EisHKKLO19}, which is a refined version of the algorithm from~\cite{KouLO18},
is~$2^{\left(\ec(A)+\td_D(A)\right)\td_D(A) 2^{\td_D(A)}}=2^{\ec(A)2^{O\left(\td_D(A)\right)}}$.

\subsection{Structure of the paper}

We now briefly describe how the paper is organized.
In Section~\ref{sec:notation}, we introduce notation used in this paper,
in particular the notions of (dual) tree-depth and branch-depth of matrices, and
provide the relevant background on matroids.
We also give the definition of an extended depth-decomposition,
which is a depth-decomposition of a matroid enhanced with information on the structure of subspaces
represented by branches of the decomposition.
Section~\ref{sec:struct1} is devoted to proving a structural result on matroids,
which establishes the existence of extended depth-decompositions with optimal depth.
The results of Section~\ref{sec:struct1} are used to prove that
the branch-depth of a matrix is equal to the minimum tree-depth of a row-equivalent matrix in Section~\ref{sec:bdtd}.
We next apply these results in Section~\ref{sec:param} to prove Theorem~\ref{thm:alg1} and Corollary~\ref{cor:ip},
which implies that integer programming is fixed parameter tractable
when parameterized by the branch-depth and the entry complexity of the constraint matrix.

The rest of the paper is devoted to computing optimal branch-depth of matroids and matrices.
As a preparation for proofs of our main results,
we develop additional tools to manipulate depth-decompositions of vector matroids in Section~\ref{sec:struct2}.
These tools are used in Section~\ref{sec:finite} to construct a dynamic programming algorithm
for computing optimal branch-depth decompositions of matroids represented over finite fields.
Finally, in Section~\ref{sec:rational},
we use the algorithm from Section~\ref{sec:finite} to design a fixed parameter algorithm
for computing optimal branch-depth decompositions of matroids represented over rationals and
for computing row-equivalent matrices with optimal branch-depth (Theorem~\ref{thm:algrational}).

\section{Notation}
\label{sec:notation}

In this section, we fix the notation used throughout the paper, present the notions of graph tree-depth and matroid branch-depth,
and include relevant results concerning them that we will need later.
To avoid our presentation becoming cumbersome through adding or subtracting one at various places,
we define the \emph{depth} of a rooted tree to be the maximum number of edges on a path from the root to a leaf, and
define the \emph{height} of a rooted tree to be the maximum number of vertices on a path from the root to a leaf,
i.e., the height of a rooted tree is always equal to its depth plus one.
The height of a rooted forest~$F$ is the maximum height of a rooted tree in~$F$.
The \emph{depth} of a vertex in a rooted tree 
is the number of edges on the path from the root to that particular vertex;
in particular, the depth of the root is zero.
The \emph{closure~$\cl(F)$} of a rooted forest $F$ is the graph obtained
by adding edges from each vertex to all its descendants.
Finally,
the \emph{tree-depth~$\td(G)$} of a graph~$G$ is the minimum height of a rooted forest~$F$ such that
the closure~$\cl(F)$ of the rooted forest~$F$ contains~$G$ as a subgraph.
It can be shown that 
the path-width of a graph~$G$ is at most its tree-depth~$\td(G)$ minus one,
and in particular, the tree-width of~$G$ is at most its tree-depth minus one (see e.g.~\cite{CygFKLMPPS15} for the definitions of path-width and tree-width).
As \cite{KarKLM17} is one of our main references, we would like to highlight that
the tree-depth as used in~\cite{KarKLM17} is equal to the minimum depth of a rooted tree~$T$ such that $G\subseteq\cl(T)$;
however, we here follow the definition of tree-depth that is standard.

The \emph{primal graph} of an~$m\times n$ matrix~$A$
is the graph~$G_P(A)$ with vertices~$\{1,\ldots,n\}$,
i.e., its vertices one-to-one correspond to the columns of~$A$,
where vertices~$i$ and $j$ are adjacent if~$A$ contains a row whose~$i$-th and $j$-th entries are non-zero.
The \emph{primal tree-depth~$\td_P(A)$} of a matrix~$A$ is the tree-depth of its primal graph.
Analogously, the \emph{dual graph} of~$A$
is the graph~$G_D(A)$ with vertices~$\{1,\ldots,m\}$,
i.e., its vertices one-to-one correspond to the rows of~$A$,
where vertices~$i$ and $j$ are adjacent if~$A$ contains a column whose~$i$-th and $j$-th entries are non-zero.
Note that the dual graph~$G_D(A)$ is isomorphic to the primal graph of the matrix transpose~$A^{T}$.
Finally,
the \emph{dual tree-depth} of~$A$, which is denoted by~$\td_D(A)$, is the tree-depth of the dual graph~$G_D(A)$.

Before introducing the notion of the branch-depth of a matroid, we review basic definitions from matroid theory. A detailed introduction to matroid theory can be found in one of the standard textbooks on the topic, e.g.~\cite{Oxl11},
but we include a brief overview of relevant concepts for completeness.
A \emph{matroid} $M$ is a pair~$(X,\II)$,
where~$\II$ is a non-empty hereditary collection of subsets of~$X$ that satisfies the \emph{augmentation axiom}.
More specifically, the collection~$\II$ is hereditary if~$\II$ contains all subsets of~$X'$ for every~$X'\in\II$, and
the augmentation axiom asserts that for all~$X'\in\II$ and $X''\in\II$ with~$|X'|<|X''|$,
there exists an element~$x\in X''\setminus X'$ such that $X'\cup\{x\}\in\II$.
The sets contained in~$\II$ are referred to as \emph{independent}.
The \emph{rank} $r(X')$ of a set~$X'\subseteq X$ is the maximum size of an independent subset of~$X'$;
the rank~$r(M)$ of a matroid~$M=(X,\II)$ is the rank of~$X$ and
an independent set of size~$r(M)$ is a \emph{basis} of~$M$.
A \emph{circuit} is a set~$X'\subseteq X$ such that $X'$ is not independent but every proper subset of~$X'$ is.
Two elements~$x$ and $x'$ of~$X$ are said to be \emph{parallel} if~$r(\{x\})=r(\{x'\})=r(\{x,x'\})=1$, and an element~$x$ is a \emph{loop} if~$r(\{x\})=0$.

Two particular examples of matroids are graphic matroids and vector matroids.
If~$G$ is a graph, then the pair~$(E(G),\II)$
where~$\II$ contains all acyclic subsets of edges of~$G$
is a matroid and is denoted by~$M(G)$; matroids of this kind are called \emph{graphic matroids}.
If~$X$ is a set of vectors of a vector space and $\II$ contains all subsets of~$X$ that are linearly independent,
then the pair~$(X,\II)$ is a matroid;
matroids of this kind are \emph{vector matroids}.
In the setting of vector matroids, the rank of~$X'\subseteq X$ is the dimension of the linear hull of~$X'$.
If~$(X,\II)$ is a vector matroid, we write~$\lin{X'}$ for the linear hull of the vectors contained in~$X'\subseteq X$ and
abuse the notation by writing~$\dim X'$ for~$\dim\lin{X'}$.

In what follows, we will need a notion of a quotient of a vector space, which we now recall.
If~$A$ is a vector space and $K$ a subspace of~$A$,
the \emph{quotient space} $A/K$ is a vector space of dimension~$\dim A-\dim K$
obtained from~$A$ by considering cosets of~$A$ given by~$K$ and
inheriting addition and scalar multiplication from~$A$;
see e.g.~\cite{Hal93} for further details.
One can show that for every subspace~$K$ of~$A$,
there exists a subspace~$B$ of~$A$ with dimension~$\dim A-\dim K$ such that
each coset contains a single vector from~$B$,
i.e., every vector~$w$ of~$A$ can be uniquely expressed as the sum of a vector~$w_B$ of~$B$ and a vector~$w_K$ of~$K$.
We call the vector~$w_B$ the \emph{quotient} of~$w$ by~$K$.
Note that the quotient of a vector is not uniquely defined by~$K$, however,
it becomes uniquely defined when the subspace~$B$, which intersects each coset at a single vector, is fixed.

If~$M=(X,\II)$ is a matroid and $X'\subseteq X$,
then the \emph{restriction} of~$M$ to~$X'$, which is denoted by~$M\left[X'\right]$, is the matroid~$(X',\II\cap 2^{X'})$.
The \emph{contraction} of~$M$ by~$X'$, which is denoted by~$M/X'$,
is the matroid with the elements~$X\setminus X'$ such that
a set~$X''\subseteq X\setminus X'$ is independent in~$M/X'$ if and only if~$r(X''\cup X')=\lvert X''\rvert+r(X')$.
Analogously, if~$M=(X,\II)$ is a vector matroid and $K$ a subspace of the linear hull of~$X$,
then the matroid~$M/K$ obtained by \emph{contracting} along the subspace~$K$ is the matroid
whose elements are the vectors of~$X$, with a subset~$X'\subseteq X$ being independent in~$M/K$ if~$X'$ is linearly independent in the quotient space~$\lin{X}/K$.
Note that if~$X'\subseteq X$, then~$M/X'$ is the matroid obtained by contracting along the linear hull of~$X'$ and
then restricting to the subset~$X\setminus X'$.

A matroid~$M$ is said to be \emph{connected} if every two (distinct) non-loop elements of~$M$ are contained in a circuit.
A set~$X'\subseteq X$ of a matroid is a \emph{component} of~$M$
if it is an inclusion-wise maximal subset of~$X$ such that the matroid~$M\left[X'\right]$ is connected and $X'$ is loop-free.
Equivalently, $X'$ is a component of~$M$
if it is an inclusion-wise minimal non-empty subset of~$X$ such that $r(X')+r(X\setminus X')=r(M)$.
If~$M$ is a vector matroid with rank at least one,
then~$M$ is connected if and only if
$M$ has no loops and
there do not exist two vector spaces~$A$ and $B$ such that $A\cap B$ contains the zero vector only,
both~$A$ and $B$ contains a non-loop element of~$M$ and
every element of~$M$ is contained in~$A$ or~$B$.

Since several algorithms presented in this paper work with matroids,
we must fix the way that the time complexity of algorithms involving matroids is measured.
All algorithms that we present work with vector matroids and
we assume that matroids are given by their vector representation.
It is then possible to use standard linear algebra algorithms to determine
which subsets of the elements of such matroids are independent.
Hence, we say that an algorithm is polynomial-time
if its running time is polynomial in the number of elements of the matroid and the complexity of its vector representation.
We remark that some of the algorithms that we use, in particular the one in Theorem~\ref{thm:approx},
are also polynomial in the more demanding setting when matroids are given by the independence oracle,
which we do not consider here.

A \emph{depth-decomposition} of a matroid~$M=(X,\II)$ is a pair~$(T,f)$,
where~$T$ is a rooted tree and $f$ is a mapping from~$X$ to the leaves of~$T$ such that
the number of edges of~$T$ is the rank of~$M$ and
the following holds for every subset~$X'\subseteq X$:
the rank of~$X'$ is at most the number of edges contained in paths from the root to the vertices~$f(x)$, $x\in X'$.
The \emph{branch-depth~$\bd(M)$} of a matroid~$M$ is the smallest depth of
a tree~$T$ that forms a depth-decomposition of~$M$.
For example, if~$M=(X,\II)$ is a matroid of rank~$r$, $T$ is a path with~$r$ edges rooted at one of its end vertices, and
$f$ is a mapping such that $f(x)$ is equal to the (non-root) leaf of~$T$ for all~$x\in X$,
then the pair~$(T,f)$ is a depth-decomposition of~$M$.
In particular, the branch-depth of any matroid~$M$ is well-defined and is at most the rank of~$M$.
We remark that the notion of matroid branch-depth given here is the one defined in~\cite{KarKLM17};
another matroid parameter, which is also called branch-depth but is different from the one that we use here,
is defined in~\cite{DevKO19}.
Finally,
the \emph{branch-depth~$\bd(A)$} of a matrix~$A$ is the branch-depth of the vector matroid formed by the columns of~$A$.
Since the vector matroid formed by the columns of a matrix~$A$ and
the vector matroid formed by the columns of any matrix row-equivalent to~$A$ are the same,
the branch-depth of~$A$ is invariant under row operations.

Similarly to the relation between tree-depth of graphs and long paths,
branch-depth is related to the existence of long circuits in a matroid~\cite[Proposition 3.4]{KarKLM17}.
\begin{proposition}
\label{prop:circuit}
Let~$M$ be a matroid.
If the branch-depth of~$M$ is~$d$,
then every circuit of~$M$ has at most~$2^d$ elements.
\end{proposition}
Kardo\v s et al.~\cite{KarKLM17} also established the following relations between the tree-depth of a graph~$G$ and
the branch-depth of the associated matroid~$M(G)$.
It is worth noting that Proposition~\ref{prop:bdtd2} does not hold without the assumption on~$2$-connectivity of a graph~$G$:
the tree-depth of an~$n$-vertex path is~$\lfloor\log_2 n\rfloor$,
however, its matroid is formed by~$n-1$ independent elements, i.e, its branch-depth is one.
\begin{proposition}
\label{prop:bdtd1}
For any graph~$G$,
the branch-depth of the graphic matroid~$M(G)$ is at most the tree-depth of the graph~$G$ decreased by one.
\end{proposition}
\begin{proposition}
\label{prop:bdtd2}
For any~$2$-connected graph~$G$,
the branch-depth of the graphic matroid~$M(G)$ is at least~$\frac{1}{2}\log_2(\td(G)-1)$.
\end{proposition}
\noindent Further properties of depth-decompositions and the branch-depth of matroids can be found in~\cite{KarKLM17}.

An \emph{extended depth-decomposition} of a vector matroid~$M=(X,\II)$ is a triple $(T,f,g)$ such that
$(T,f)$ is a depth-decomposition of~$M$ and
$g$ is a bijective mapping from the non-root vertices of~$T$ to a basis of the linear hull of~$X$ that satisfies that
every element~$x\in X$ is contained in the linear hull of the~$g$-image of the non-root vertices on the path from~$f(x)$ to the root of~$T$.
Note that $g$-images need not be elements of~$M$ in general;
if all~$g$-images are elements of the matroid~$M$,
we say that an extended depth-decomposition~$(T,f,g)$ is \emph{principal}.
Kardo\v s et al.~\cite[Corollary 3.17]{KarKLM17} designed an algorithm that
outputs an approximation of an optimal depth-decomposition, which is a principal extended depth-decomposition;
we state the result here for the case of vector matroids.
\begin{theorem}
\label{thm:approx}
There exists a polynomial-time algorithm that given a vector matroid~$M$ and an integer~$d$,
either outputs that the branch-depth of~$M$ is larger than~$d$ or
outputs a principal extended depth-decomposition of~$M$ of depth at most~$4^d$.
\end{theorem}
If~$(T,f,g)$ is an extended depth-decomposition and $u$ is a vertex of~$T$,
then~$K_u$ is the linear hull of the~$g$-images of the vertices on the path from~$u$ to the root of~$T$;
in particular, if~$u$ is the root, then~$K_u$ contains the zero vector only.
It will always be clear from the context for which extended depth-decomposition of~$M$ the spaces~$K_u$ are defined,
in particular, the vertex~$u$ determines which rooted tree~$T$ is considered.

\begin{figure}
\begin{center}
\epsfbox{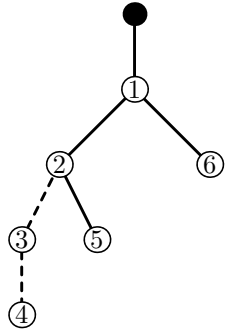}
\end{center}
\caption{An illustration of the definition of solid branches and branches at capacity.
         The picture depicts the tree~$T$ of an extended depth-decomposition~$(T,f,g)$ such that
	 the function~$g$ maps a non-root vertex labeled with~$i$ to the~$i$-th unit vector.
	 If the~$f$-preimage of the leaf of the branch of~$T$ depicted by dashed edges
	 is~$\{(1,1,0,0,0,0),(0,1,1,0,0,0),(0,0,1,0,0,0),(0,0,0,1,0,0)\}$,
	 then the branch is at capacity (regardless of the structure of the matroid and the choice of~$f$);
	 however, if the~$f$-preimage is~$\{(0,0,1,0,0,0),(0,0,0,1,0,0)\}$,
	 then the branch is not at capacity.
	 If the~$f$-preimage is~$\{(0,0,1,1,0,0),(0,1,1,0,0,0),(1,0,0,1,0,0)\}$, then the branch is solid, and
	 if the~$f$-preimage is~$\{(0,0,1,0,0,0),(1,0,1,0,0,0),(1,0,0,1,0,0)\}$, then the branch is not solid.
         }
\label{fig:branches}
\end{figure}

A \emph{branch} of a rooted tree~$T$
is a subtree~$S$ rooted at a vertex~$u$ of~$T$ such that $u$ has at least two children, and
the subtree~$S$ contains exactly~$u$, one child~$u'$ of~$u$, and all descendants of~$u'$.
In particular, a rooted tree has a branch if and only if it has a vertex with at least two children.
A branch~$S$ is \emph{primary} if every ancestor of the root of~$S$ has exactly one child.
Every rooted tree~$T$ that is not a rooted path has at least two primary branches and
all primary branches are rooted at the same vertex.
Let~$(T,f)$ be a depth-decomposition of a matroid~$M$.
We write~$\widehat{S}$ for the set of elements of the matroid~$M$ mapped by~$f$ to the leaves of~$S$ and
$\|S\|$ for the number of edges of~$S$.
Let~$S$ be a branch of~$T$ and $S_1,\ldots,S_k$ be the other branches with the same root.
The branch~$S$ is \emph{at capacity} if
\[r\left(X\setminus\left( \widehat{S}_1\cup\cdots\cup \widehat{S}_k \right) \right)=r(M)-\|S_1\|-\|S_2\|-\cdots-\|S_k\|,\]
where~$X$ is the set of all elements of the matroid~$M$ (see Figure~\ref{fig:branches} for an illustration).
Note that if~$S$ is primary,
then the left side of the equality is~$r(\widehat{S})$ and the right side is~$h+\|S\|$,
where~$h$ is the depth of the root of~$S$.
In particular, a primary branch~$S$ is \emph{at capacity} if and only if the rank of~$\widehat{S}$ is equal to the sum of~$\|S\|$ and $h$,
i.e., if and only if the rank inequality from the definition of a depth-decomposition holds with equality for the set~$\widehat{S}$.
Finally, 
if~$(T,f,g)$ is an extended depth-decomposition of~$M$,
we say that a branch~$S$ rooted at a vertex~$u$ is \emph{solid}
if the matroid~$(M/K_u)\left[\widehat{S}\right]$ after removal of its loops is connected, and that $(T,f,g)$ is solid if all of its branches are;
again, an illustration can be found in Figure~\ref{fig:branches}.

\section{Optimal extended depth-decompositions}
\label{sec:struct1}

The goal of this section is to show that every vector matroid has an extended depth-decomposition
with depth equal to its branch-depth.
To do so, we start with showing that branches rooted at the root of a decomposition tree are always at capacity.

\begin{lemma}
\label{lm:rootbranch}
Let~$(T,f)$ be a depth-decomposition of a vector matroid~$M$.
If~$T$ has a branch~$S$ rooted at the root of~$T$, then~$S$ is at capacity.
\end{lemma}

\begin{proof}
Suppose that a branch~$S$ rooted at the root of~$T$ is not at capacity.
This implies that $\dim\widehat{S}<\|S\|$.
Let~$X'$ be the set of elements of~$M$ that are not contained in~$\widehat{S}$.
By the definition of a depth-decomposition, $\dim X'$ is at most~$r(M)-\|S\|$.
However, the submodularity of the dimension implies that
$\dim\widehat{S}\cup X'<r(M)$, which is impossible.
\end{proof}

The following lemma is a core of our argument that
every matroid has a depth-decomposition of optimal depth such that
each primary branch is at capacity.
An illustration of the operation described in the statement of the lemma is given in Figure~\ref{fig:primbranch}.

\begin{lemma}
\label{lm:primbranch}
Let~$(T,f)$ be a depth-decomposition of a vector matroid~$M$.
Assume that $T$ contains a primary branch~$S$ that is not at capacity.
Let~$u$ be the root of~$S$, and
let~$T'$ be the rooted tree obtained from~$T$ by changing the root of~$S$ to be the parent of~$u$.
Then, $(T',f)$ is a depth-decomposition of~$M$.
\end{lemma}

\begin{figure}
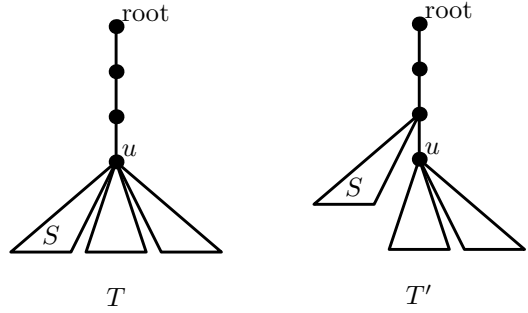

\begin{center}
\epsfbox{bdepth_ip-1.mps}
\hskip 10mm
\epsfbox{bdepth_ip-2.mps}
\end{center}
\caption{The trees~$T$ and $T'$ from the statement of Lemma~\ref{lm:primbranch}.}
\label{fig:primbranch}
\end{figure}

\begin{proof}
By Lemma~\ref{lm:rootbranch}, $u$ has a parent and thus~$T'$ is well-defined. Let~$X$ be the set of elements of~$M$ and fix a subset~$X'$ of~$X$.
We need to show that $\dim X'$ is at most the number~$e_0$ of edges on the paths in~$T'$
from the vertices in the~$f$-image of~$X'$ to the root.
If~$X'$ contains an element of~$X\setminus\widehat{S}$,
then the number of such edges is the same in the trees~$T$ and $T'$ and
the inequality follows from the fact that $(T,f)$ is a depth-decomposition of~$M$.
Hence, we will assume that $X'$ is a subset of~$\widehat{S}$.
Observe that collectively the primary branches of~$T$ different from~$S$ contain~$r(M)-h-\|S\|$ edges,
where~$h$ is the depth of~$u$.
We derive using the fact that $(T,f)$ is a depth-decomposition the following:
\begin{align*}
e_0+1+(r(M)-h-\|S\|) & \ge \dim X'\cup (X\setminus\widehat{S}) \\
                     & = \dim X'+\dim X\setminus\widehat{S} - \dim \lin{X'}\cap\lin{X\setminus\widehat{S}} \\
		     & \ge \dim X'+\dim X\setminus\widehat{S} - \dim \lin{\widehat{S}}\cap\lin{X\setminus\widehat{S}} \\
		     & = \dim X'+\dim X\setminus\widehat{S} -(\dim \widehat{S}+\dim X\setminus\widehat{S}-\dim X) \\
		     & = \dim X'-\dim \widehat{S}+r(M).
\end{align*}
This implies that $\dim X'$ is at most
\[e_0+\dim\widehat{S}+1-h-\|S\|\le e_0,\]
where the inequality follows using that $S$ is not at capacity,
i.e., $\dim\widehat{S}<h+\|S\|$.
Hence, $(T',f)$ is a depth-decomposition of~$M$.
\end{proof}

We can now show that
every matroid has a depth-decomposition of optimal depth such that
each primary branch is at capacity.
We state the next two lemmas and the theorem that follows them for an arbitrary depth-decomposition~$(T,f)$,
i.e., a depth-decomposition of not necessarily optimal depth,
since we will need to apply the algorithmic arguments used in their proofs for arbitrary depth-decompositions later.

\begin{lemma}
\label{lm:capacity}
Let~$(T,f)$ be a depth-decomposition of a vector matroid~$M=(X,\II)$ of depth~$d$.
There exists a depth-decomposition of~$M$ of depth at most~$d$ such that every primary branch is at capacity.
\end{lemma}

\begin{proof}
If~$T$ is a rooted path, then the lemma holds vacuously.
Suppose that $T$ is not a rooted path.
If all primary branches of~$(T,f)$ are at capacity, then we are done.
If not, we consider a primary branch of~$T$ that is not at capacity and
apply Lemma~\ref{lm:primbranch} to obtain a depth-decomposition~$(T',f)$.
If all primary branches of~$(T',f)$ are at capacity, then we are done.
If not, we consider a primary branch of~$T'$ that is not at capacity and iterate the process.
Note that at each iteration, the sum of the lengths of the paths from the leaves to the root decreases,
so the process eventually stops with a depth-decomposition such that all its primary branches are at capacity.
\end{proof}

Now we analyze depth-decompositions whose primary branches are at capacity.

\begin{lemma}
\label{lm:primintersect}
Let~$(T,f)$ be a depth-decomposition of a vector matroid~$M=(X,\II)$ such that 
$T$ is not a rooted path and each primary branch of~$T$ is at capacity.
Let~$S_1,\ldots,S_k$ be the primary branches of~$T$, and
let~$A_1,\ldots,A_k$ be the linear hulls of~$\widehat{S_1},\ldots,\widehat{S_k}$, respectively.
Further, let~$h$ be the depth of the common root of~$S_1,\ldots,S_k$ in~$T$.
There exists a subspace~$K$ of dimension~$h$ such that $A_i\cap A_j=K$ for all~$1\le i<j\le k$.
\end{lemma}

\begin{proof}
Consider~$i$ and $j$ such that $1\le i<j\le k$.
Since~$S_i$ is at capacity, we obtain that $\dim \widehat{S_i}=\dim A_i=h+\|S_i\|$.
Analogously, it holds that $\dim \widehat{S_j}=\dim A_j=h+\|S_j\|$.
Since~$(T,f)$ is a depth-decomposition, we deduce that
\begin{align*}
h+\|S_i\|+\|S_j\| & \geq \dim A_i\cup A_j \\
                  & = \dim A_i + \dim A_j-\dim A_i\cap A_j\\
		  & = (h+\|S_i\|)+(h+\|S_j\|)-\dim A_i\cap A_j,
\end{align*}		  
which implies that $\dim A_i\cap A_j\ge h$.
On the other hand, it holds that
\begin{align*}
\dim A_i\cap A_j & \le \dim A_i\cap\lin{\bigcup_{j'\not=i}A_{j'}} \\
                 & = \dim A_i + \dim \bigcup_{j'\not=i}A_{j'} - \dim A_i\cup\bigcup_{j'\not=i}A_{j'}\\
		 & = h+\|S_i\|+\dim \bigcup_{j'\not=i}A_{j'}-h-\sum_{j'=1}^k\|S_{j'}\|\\
		 & \le h+\|S_i\|+h+\sum_{j'\not=i}\|S_{j'}\|-h-\sum_{j'=1}^k\|S_{j'}\|=h.
\end{align*}
We conclude that $\dim A_i\cap A_j=h$.
This implies that the first inequality in the expression above holds with equality,
so
\begin{equation}
A_i\cap\lin{\bigcup_{j'\not=i}A_{j'}}
=A_i\cap A_j
=A_j\cap A_i
=A_j\cap\lin{\bigcup_{i'\not=j}A_{i'}},\label{eq:intersect}
\end{equation}
where the last step follows by the same argument above with indices~$i$ and $j$ swapped. Since the left hand side of \eqref{eq:intersect} is independent of the choice of~$j$ and the right hand side is independent of the choice of~$i$, it must hold that $A_{i}\cap A_{j}$ is independent of the choice of both~$i$ and $j$.
This intersection forms the required space~$K$ of dimension~$h$.
\end{proof}

We are now ready to prove the main theorem of this section.

\begin{theorem}
\label{thm:extended}
Let~$(T,f)$ be a depth-decomposition of a vector matroid~$M=(X,\II)$ of depth~$d$.
There exists an extended depth-decomposition of~$M$ of depth at most~$d$.
\end{theorem}

\begin{proof}
The proof proceeds by induction on the rank of~$M$.
If~$T$ is a rooted path,
we assign elements of a basis of~$\lin{X}$ to the non-root vertices of~$T$ arbitrarily,
i.e., we choose~$g$ to be any bijection to a basis of~$\lin{X}$,
which yields an extended depth-decomposition~$(T,f,g)$ of~$M$.
Note that if the rank of~$M$ is one, then~$T$ is the one-edge rooted path,
i.e., this case covers the base of the induction in particular.

We assume that $T$ is not a rooted path for the rest of the proof.
By Lemma~\ref{lm:capacity}, we can assume that all primary branches of~$T$ are at capacity.
Let~$S_1,\ldots,S_k$ be the primary branches of~$T$, and
let~$h\ge 0$ be the depth of the common root of~$S_1,\ldots,S_k$.
By Lemma~\ref{lm:primintersect}, there exists a subspace~$K$ of dimension~$h$ such that
the intersection of linear hulls of~$\widehat{S_i}$ and $\widehat{S_j}$ is~$K$ for all~$1\le i<j\le k$;
let~$b_1,\ldots,b_h$ be an arbitrary basis of~$K$.

We define~$M_i$, $i=1,\ldots,k$, to be the matroid such that
the elements of~$M_i$ are~$\widehat{S_i}$ and $X'\subseteq\widehat{S_i}$ is independent
if and only if the elements~$X'\cup\{b_1,\ldots,b_h\}$ are linearly independent.
In particular, the rank of~$X'\subseteq\widehat{S_i}$ in~$M_i$ is equal to~$\dim X'\cup K-h$.
The matroid~$M_i$ can be viewed as obtained by taking the vector matroid with the elements~$\widehat{S_i}\cup\{b_1,\ldots,b_h\}$ and
contracting the elements~$b_1,\ldots,b_h$. In particular, $M_i$ is a vector matroid, and
the vector representation of~$M_i$ can be obtained from~$\widehat{S_i}$ by taking quotients by~$K$.
Note that the rank of~$M_i$ is~$\dim\widehat{S_i}\cup K-h$,
i.e., its rank is smaller than the rank of~$M$ and we will be able to eventually apply induction to it.

Let~$f_i$ be the restriction of~$f$ to~$\widehat{S_i}$.
We claim that $(S_i,f_i)$ is a depth-decomposition of~$M_i$.
Let~$X'$ be a subset of~$\widehat{S_i}$, and
let~$e_i$ be the number of edges contained in the union of paths from the elements~$f_i(x)$, $x\in X'$, to the root of~$S_i$.
By the definition of~$M_i$, the rank of~$X'$ in~$M_i$ is equal to~$\dim X'\cup K-h$.
Choose an arbitrary~$j\not=i$, $1\le j\le k$.
Since the intersection of linear hulls of~$\widehat{S_i}$ and $\widehat{S_j}$ is~$K$,
$(T,f)$ is a depth-decomposition of~$M$, and
the branch~$S_j$ is at capacity, i.e., $\dim\widehat{S_j}=\|S_j\|+h$,
we obtain that the rank of~$X'$ in~$M_i$ is equal to
\begin{align*}
\dim X'\cup K-h & = \dim X'\cup\widehat{S_j}-\dim\widehat{S_j} \\
 & \le e_i+\|S_j\|+h-\dim\widehat{S_j}=e_i.
\end{align*}
Hence, $(S_i,f_i)$ is a depth-decomposition of~$M_i$.

We apply induction to each matroid~$M_i$ and its depth-decomposition~$(S_i,f_i)$, $i=1,\ldots,k$,
to obtain extended depth-decompositions~$(S'_i,f'_i,g_i)$ of~$M_i$ such that the depth of~$S'_i$ is at most the depth of~$S_i$.
Let~$T'$ be a rooted tree obtained from a rooted path of length~$h$
by identifying its non-root end with the roots of~$S'_1,\ldots,S'_k$.
Observe that the depth of~$T'$ does not exceed the depth of~$T$.
Further, let~$f'$ be the unique function from~$X$ to the leaves of~$T$ such that
the restriction of~$f'$ to the elements of~$M_i$ is~$f'_i$.
Finally, let~$g$ be any function from the non-root vertices of~$T$ such that 
the~$h$ non-root vertices of the path from the root are mapped by~$g$ to the vectors~$b_1,\ldots,b_h$ by~$g$ and
$g(v)=g_i(v)$ for every non-root vertex~$v$ of~$S_i$.

We claim that $(T',f',g)$ is an extended depth-decomposition of~$M$.
We first verify that, for every~$x\in X$, $f'(x)$ is contained in the linear hull of the~$g$-image
of the non-root vertices on the path from~$f'(x)$ to the root.
Fix~$x\in X$ and let~$i$ be such that $x\in\widehat{S_i}$.
Since~$(S'_i,f'_i,g_i)$ is an extended depth-decomposition of~$M_i$,
$x$ is contained in the linear hull of~$K$ and
the~$g_i$-images of the non-root vertices on the path from~$f'(x)=f_i(x)$ to the root of~$S'_i$.
Hence, $x$ is contained 
in the linear hull of the~$g$-image of the non-root vertices on the path from~$f'(x)$ to the root of~$T'$.

Consider now an arbitrary subset~$X'\subseteq X$.
We have already established that
all elements of~$X'$ are contained in the linear hull of the~$g$-image of the non-root vertices
on the paths from~$f'(x)$, $x\in X'$, to the root of~$T'$.
Since the dimension of this linear hull is at most the number of non-root vertices on such paths,
which is equal to the number of edges on the paths, it follows that $(T',f')$ is a depth-decomposition of~$M$.
\end{proof}

\section{Optimal tree-depth of a matrix}
\label{sec:bdtd}

In this section, we relate the optimal dual tree-depth of a matrix~$A$ to its branch-depth.
We start with showing that the branch-depth of a matrix~$A$ is at most its dual tree-depth.

\begin{proposition}
\label{prop:bdtd}
If~$A$ is an~$m\times n$ matrix, then~$\bd(A)\le\td_D(A)$.
\end{proposition}

\begin{proof}
We assume without loss of generality that the rows of the matrix~$A$ are linearly independent.
Indeed, deleting a row of~$A$ that can be expressed as a linear combination of other rows of~$A$
does not change the structure of the matroid formed by the columns of~$A$. In particular,
the branch-depth of~$A$ is preserved by deleting such a row, and
the deletion cannot increase the tree-depth of the dual graph~$G_D(A)$ (the dual graph of the new matrix
is a subgraph of the original dual graph and the tree-depth is monotone under taking subgraphs).

Let~$X$ be the set of rows of the matrix~$A$ and $Y$ the set of its columns.
Further,
let~$T$ be a rooted forest of height~$\td_D(A)$ with the vertex set~$X$ such that
its closure contains the dual graph~$G_D(A)$ as a subgraph.
Consider the rooted tree~$T'$ obtained from~$T$ by adding a new vertex~$w$,
making~$w$ adjacent to the roots of all trees in~$T$ and also making~$w$ to be the root of~$T'$.
Since the rows of~$A$ are linearly independent, the number of edges of~$T'$ is equal to the row rank of~$A$,
which is the same as its column rank.
In particular, the number of edges of~$T'$ is the rank of the vector matroid formed by the columns of~$A$.

We next define a function~$f:Y\to V(T')$ such that
the pair~$(T',f)$ is a depth-decomposition of the vector matroid formed by the columns of~$A$. Let~$y$ be a column of~$A$, and observe that all rows~$x$ such that the entry in the row~$x$ and the column~$y$ is non-zero
form a complete subgraph of the dual graph~$G_D(A)$.
Hence they must lie on some path from a leaf to the root of~$T'$;
set~$f(y)$ to be any such leaf.

Since the depth of~$T'$ is~$\td_D(A)$ (the height of~$T'$ is~$\td_D(A)+1$),
the proof will be completed by showing that $(T',f)$ is a depth-decomposition of
the vector matroid formed by the columns of~$A$.
Consider a subset~$Y'\subseteq Y$ of columns of~$A$ and
let~$X'$ be the set of rows (vertices of the dual graph) on the path from~$f(y)$ for some~$y\in Y'$ to the root of~$T'$.
Note that $|X'|$ is equal to the number of edges contained in such paths.
The definition of~$f$  yields that
every column~$y\in Y'$ has non-zero entries only in the rows~$x$ such that $x\in X'$.
Hence, the rank of~$Y'$ is at most~$|X'|$.
It follows that the pair~$(T',f)$ is a depth-decomposition of the vector matroid formed by the columns of~$A$.
\end{proof}

We next prove the main theorem of this section.

\begin{theorem}
\label{thm:bdbasis}
Let~$A$ be an~$m\times n$ matrix of rank~$m$,
let~$M$ be the vector matroid formed by columns of~$A$, and
let~$(T,f,g)$ be an extended depth-decomposition of~$M$.
Further, let~$\image(g)=\{w_1,\ldots,w_m\}$.
The dual tree-depth of the~$m\times n$ matrix~$A'$ such that
the~$j$-th column of~$A$ is equal to
\[\sum_{i=1}^m A'_{ij}w_i\]
is at most the depth of the tree~$T$.
\end{theorem}

\begin{proof}
Let~$F$ be the rooted forest obtained from~$T$ by removing the root and
associate the~$i$-th row of~$A'$ with the vertex~$v$ of~$F$ such that $g(v)=w_i$.
Note that the height of~$F$ is the depth of~$T$.
We will establish that the dual graph~$G_D(A')$ is contained in the closure~$\cl(F)$ of the forest~$F$.
Let~$i$ and $i'$, $1\le i<i'\le m$, be such that
the vertices of~$F$ associated with the~$i$-th and $i'$-th rows of~$A'$ are adjacent in~$G_D(A')$.
This means that there exists~$j$, $1\le j\le n$, such that $A'_{ij}\not=0$ and $A'_{i'j}\not=0$.
Let~$v$ be the leaf of~$T$ that is the~$f$-image of the~$j$-th column of~$A$.
The definition of an extended depth-decomposition yields that
the~$j$-th column is a linear combination of the~$g$-image of the non-root vertices on the path from~$v$ to the root of~$T$.
In particular, the path contains the two vertices of~$T$ mapped by~$g$ to~$w_i$ and $w_{i'}$;
these two vertices are associated with the~$i$-th and $i'$-th rows of~$A'$.
Hence, the vertices associated with the~$i$-th and $i'$-th rows are adjacent in~$\cl(F)$.
We conclude that $G_D(A')$ is a subgraph of~$\cl(F)$.
\end{proof}

We are now ready to prove Theorem~\ref{thm:equal}.

\begin{proof}[Proof of Theorem~\ref{thm:equal}]
Let~$\td_D^*(A)$ be the smallest dual tree-depth of a matrix that is row-equivalent to~$A$.
By Proposition~\ref{prop:bdtd}, it holds that $\bd(A)\le\td_D^*(A)$.
We now prove the other inequality.
We can assume without loss of generality that the rank of~$A$ is equal to the number of its rows;
if this is not the case,
we can apply the following arguments to the matrix~$A$ restricted to a maximal linearly independent set of rows and
then use row operations to make all entries of the remaining rows to be equal to zero.
Since rows with all entries equal to zero correspond to isolated vertices in the dual graph,
their presence does not affect the dual tree-depth of the matrix.

Let~$M$ be the vector matroid formed by the columns of~$A$.
By Theorem~\ref{thm:extended}, the matroid~$M$ has an extended depth-decomposition~$(T,f,g)$ with depth~$\bd(A)=\bd(M)$.
Let~$A'$ be the matrix from the statement of Theorem~\ref{thm:bdbasis}.
Note that $A'=B^{-1}A$
where~$B$ is the~$m\times m$ matrix such that
the columns of~$B$ are the vectors~$w_1,\ldots,w_m$ from the statement of Theorem~\ref{thm:bdbasis}.
In particular, $A'$ is row-equivalent to~$A$ since~$\{w_1,\ldots,w_m\}$ is a basis.
As~$\td_D(A')$ is at most the depth of~$T$, which is~$\bd(A)$, it follows that $\td_D^*(A)\le\bd(A)$.
\end{proof}

\section{Algorithms for integer programming}
\label{sec:param}

The main purpose of this section is to combine Theorem~\ref{thm:equal}
with the existing approximation algorithm for branch-depth (Theorem~\ref{thm:approx})
to obtain an approximation algorithm for computing a row-equivalent matrix with small dual tree-depth (if it exists).

\begin{proof}[Proof of Theorem~\ref{thm:alg1}]
Let~$A$ be an~$m\times n$ matrix.
Without loss of generality, we can assume that the rows of~$A$ are linearly independent,
i.e., the rank of~$A$ is~$m$.
This also implies that the rank of the column space of~$A$ is~$m$, in particular, $n\ge m$.
We apply the approximation algorithm described in Theorem~\ref{thm:approx}
to the vector matroid~$M$ formed by the columns of the matrix~$A$, and
obtain an extended depth-decomposition~$(T,f,g)$ of~$M$.
If the depth of~$T$ is larger than~$4^d$, then the branch-depth of~$A$ is larger than~$d$;
we report this and stop.
Let~$B_g$ be the matrix with the columns formed by the vectors in~$\image(g)$ and let~$B=B_g^{-1}$.
Note that the matrix~$A'$ from the statement of Theorem~\ref{thm:bdbasis} is equal to~$BA$.
By Theorem~\ref{thm:bdbasis}, the dual tree-depth of~$A'$ is at most~$4^d$.

We will next show that the entry complexity of~$A'$ is at most~$O(d\cdot 4^d\cdot\ec(A))$.
Note that the classical implementation of the Gaussian elimination in strongly polynomial time by Edmonds~\cite{Edm67}
yields that the entry complexity of the matrix~$B$ is~$O(m\log m\cdot \ec(A))$ and
this estimate is not sufficient to bound the entry complexity of~$A'$ in the way that we need.
Let~$x$ be a column of~$A$, and
let~$W$ be the set of indices~$i$, $1\le i\le m$, such that
the~$i$-th column of~$B_g$ is~$g(v)$ for some non-root vertex~$v$ on the path from~$f(x)$ to the root of~$T$.
Note that $|W|\le 4^d$ since the depth~$T$ is at most~$4^d$.
Since the column~$x$ is a linear combination of the~$g$-images of non-root vertices on the path from~$f(x)$ to the root of~$T$,
the~$i$-th entry of the column of~$A'$ that corresponds to~$x$ is zero if~$i\not\in W$.
The remaining~$|W|$ entries of this column of~$A'$ form a solution of the following system of at most~$4^d$ linear equations:
the system is given by a matrix obtained from~$B_g$ by restricting~$B_g$ to the columns with indices in~$W$ and
to~$|W|$ rows such that the resulting matrix has rank~$|W|$, and
the right hand side of the system is formed by the entries of the column~$x$ in~$A$ corresponding to these~$|W|$ rows.
It follows (using Cramer's rule for solving systems of linear equations involving determinants) that
a solution of this system has entry complexity at most~$O(\log (4^d)!\cdot\ec(A))=O(d\cdot 4^d\cdot\ec(A))$.
Hence, the entry complexity of the matrix~$A'=BA$,
after dividing the numerator and the denominator of each entry by their greatest common divisor,
is~$O(d\cdot 4^d\cdot\ec(A))$.
\end{proof}

As explained in Section~\ref{sec:intro}, Theorem~\ref{thm:alg1} yields Corollary~\ref{cor:ip},
which asserts that integer programming is fixed parameter tractable
when parameterized by the branch-depth and the entry complexity of the constraint matrix.

\begin{proof}[Proof of Corollary~\ref{cor:ip} using Theorem~\ref{thm:alg1}]
Consider an integer program as in \eqref{IP} that has branch-depth at most~$d$.
We apply the algorithm from Theorem~\ref{thm:alg1}
to obtain a rational matrix~$B$ such that
the instance with~$A'=BA$, $\veb'=B\veb$, $\vel'=\vel$ and $\veu'=\veu$
has dual tree-depth~$D$ at most~$4^d$ in case of Theorem~\ref{thm:alg1}.

To apply the algorithm from~\cite{EisHKKLO19}, we need to transform the matrix~$A'$ into an integer matrix.
We do so by multiplying each row by the least common multiple of the denominators of the fractions in this row.
Since all denominators are integers between~$1$ and $2^{\ec(A')}$,
the value of this least common multiple is at most~$2^{\ec(A') 2^{\ec(A')}}$.
Hence, the entry complexity of the resulting integer matrix~$A''$ is~$O\left(\ec(A')2^{\ec(A')}\right)$.
Since the dependence of the algorithm from~\cite{EisHKKLO19}
on the entry complexity~$\ec(A'')$ and the dual tree-depth~$D$ of~$A''$ is~$2^{\left(\ec(A'')+D\right)D 2^{D}}$,
we obtain that the running time of the resulting algorithm depends on~$\ec(A)$ and $d$ as given in Table~\ref{tab:complexity}.
\end{proof}

We complement Corollary~\ref{cor:ip} by showing that
integer programming is not fixed parameter tractable when parameterized by the ``primal'' branch-depth.

\begin{proposition}
\label{prop:bdp}
Integer programming is \NPh for instances with constraint matrices~$A$ satisfying~$\bd(A^{T})=1$ and $\ec(A)=1$,
i.e., for instances such that the vector matroid formed by rows of the constraint matrix has branch-depth one.
\end{proposition}

\begin{proof}
An integer program as in \eqref{IP} such that the rows of the matrix~$A$ are not linearly independent
is equivalent to an integer program with a matrix~$A'$ obtained from~$A$
by a restriction to a maximal linearly independent set of rows
unless the rank of the matrix~$A$ with the column~$\veb$ added is larger than the rank of~$A$;
in the latter case, the integer program is infeasible.
Hence, it is possible in polynomial time to either determine that the input integer program
is infeasible or to find an equivalent integer program such that the rows of the constraint matrix are
linearly independent and the matrix is a submatrix of the original constraint matrix.
However, the branch-depth of the matroid formed by rows of such a (non-zero) matrix is one.
Since integer programming is already \NPh for instances such that all the entries of
the constraint matrix are~$0$ or~$\pm 1$, cf.~\cite[Proposition 101, part 2]{EisHKKLO19},
the proposition follows.
\end{proof}

\section{Structure of extended depth-decompositions}
\label{sec:struct2}

In this section, we present structural results on extended depth-decompositions that
we need to design a fixed parameter algorithm to compute a depth-decompo\-si\-tion of a vector matroid with an optimal depth.
We start with the following lemma, which can be viewed as a generalization of Lemma~\ref{lm:primintersect};
indeed, if the set~$U$ in the statement contains only the common root of the primary branches,
then the statement of the lemma is the same as that of Lemma~\ref{lm:primintersect}.

\begin{lemma}
\label{lm:allintersect}
Let~$(T,f)$ be a depth-decomposition of a vector matroid~$M$ and
let~$U$ be a set of vertices of~$T$ such that
every vertex contained in~$U$ has at least two children and
every ancestor~$u'$ of a vertex in~$U$ such that $u'$ has at least two children is also contained in~$U$.
Assume that every branch of~$T$ rooted at a vertex from~$U$ is at capacity.

Then, every vertex~$u\in U$ can be associated with a subspace~$L_u$ of the linear hull of the elements of~$M$ such that
the dimension of~$L_u$ is the depth of~$u$ and the following holds.
Let~$S_1,\ldots,S_k$ be all branches rooted at~$u$.
If each ancestor of~$u$ has a single child, let~$L_0$ be the vector space containing the zero vector only;
otherwise, let~$u'$ be the nearest ancestor of~$u$ with at least two children, and
let~$L_0$ be the space~$L_{u'}$.
It holds that
\[\dim\widehat{S_i}\cup L_0=\|S_i\|+\dim L_u\qquad\mbox{and}\qquad
  \lin{\widehat{S_i}\cup L_0}\cap\lin{\widehat{S_j}\cup L_0}=L_u\]
for all~$1\leq i < j \leq k$. In particular, $L_0\subseteq L_u$.
\end{lemma}

\begin{proof}
We proceed by induction on the size of~$U$.
If~$U$ is empty, the lemma vacuously holds.
Suppose that $\lvert U\rvert=1$ and let~$u$ be the only vertex contained in~$U$.
By the assumption of the lemma, every branch rooted at~$u$ is primary.
Hence, the statement of the lemma is implied by Lemma~\ref{lm:primintersect} and
the fact that each branch rooted at~$u$ is at capacity.

\begin{figure}
\begin{center}
\epsfbox{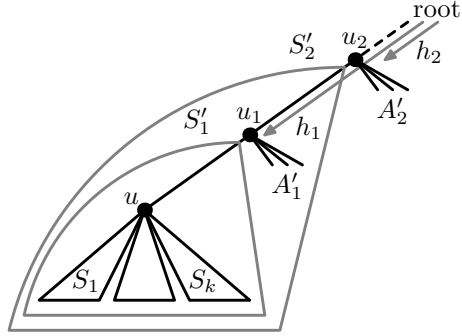}
\end{center}
\caption{Notation used in the proof of Lemma~\ref{lm:allintersect}.}
\label{fig:allintersect}
\end{figure}

Suppose that $\lvert U\rvert\ge 2$, and let~$u$ be any vertex of~$U$ with no descendant in~$U$.
Apply induction to the set~$U\setminus\{u\}$ to get subspaces~$L_{u'}$, $u'\in U\setminus\{u\}$,
with properties given in the statement of the lemma.
Let~$S_1,\ldots,S_k$ be all branches rooted at~$u$, and
let~$A_1,\ldots,A_k$ be the linear hulls of~$\widehat{S_1},\ldots,\widehat{S_k}$, respectively.
Further, let~$u_1,\ldots,u_{\ell}$ be all the vertices with at least two children on the path from the parent of~$u$ to the root (in this order), and
let~$h_1,\ldots,h_{\ell}$, be the depth of~$u_1,\ldots,u_{\ell}$, respectively.
See Figure~\ref{fig:allintersect} for an illustration of the notation.
Note that $\{u_1,\ldots,u_{\ell}\}\subseteq U$;
the choice of~$u$ implies that $\ell\ge 1$.
Let~$S'_i$, $i=1,\ldots,\ell$, be the branch rooted at~$u_i$ that contains~$u$ and
$A'_i$ be the linear hull of the elements assigned to leaves of the branches rooted at~$u_i$ different from~$S'_i$.
Note that $\widehat{S'_1}=\widehat{S_1}\cup\cdots\cup\widehat{S_k}$.

Set~$L_i=L_{u_i}$ for~$i=1,\ldots,\ell$ and also set~$L_{\ell+1}$ to be the vector space containing the zero vector only.
Note that $L_{\ell+1}\subseteq L_\ell\subseteq\cdots\subseteq L_1$ and
the dimension of~$L_i$ is~$h_i$ for~$i=1,\ldots,\ell$;
the space~$L_1$ is the vector space~$L_0$ from the statement of the lemma with respect to~$u$.
The following holds by the induction assumption for every~$i=1,\ldots,\ell$:
if~$R$ is a branch rooted at the vertex~$u_i$, then
\begin{equation}
\dim \widehat{R}\cup L_{i+1}=\dim L_i+\|R\|\label{eq:allintersect-ind-dim}
\end{equation}
and if~$R'$ is another branch rooted at the vertex~$u_i$, then
\begin{equation}
\lin{\widehat{R}\cup L_{i+1}}\cap\lin{\widehat{R'}\cup L_{i+1}}=L_i.\label{eq:allintersect-ind}
\end{equation}
The identity \eqref{eq:allintersect-ind-dim} for~$i=1$ and $R=S'_1$ yields that
\begin{equation}
\dim L_2\cup\widehat{S'_1}=\dim L_1+\|S'_1\|.\tag{\ref{eq:allintersect-ind-dim}.1}\label{eq:allintersect1D}
\end{equation}
Using \eqref{eq:allintersect-ind} iteratively for~$i=\ell,\ldots,1$,
we obtain the following; the iterative argument uses that $\widehat{S'_i}\cup A'_i\subseteq\lin{\widehat{S'_{i+1}}}$.
\begin{align}
L_i\subseteq \lin{A'_i\cup L_{i+1}}\subseteq \dots \subseteq \lin{A'_i\cup\cdots\cup A'_{\ell}} \label{eq:allintersect0A}\\
\lin{\widehat{S'_i}\cup L_{i+1}}\cap\lin{A'_i\cup L_{i+1}}=L_i \nonumber\\
\lin{\widehat{S'_i}\cup L_{i+1}}\cap\lin{A'_i\cup\cdots\cup A'_{\ell}}=L_i \label{eq:allintersect0B}
\end{align}
Let~$A'$ be the linear hull of~$A'_i\cup\cdots\cup A'_{\ell}$. The relations \eqref{eq:allintersect0A} and \eqref{eq:allintersect0B} yield for~$i=1$ the following.
\begin{align}
L_1\subseteq \lin{A'_1\cup\cdots\cup A'_\ell}=A' \tag{\ref{eq:allintersect0A}.1}\label{eq:allintersect1A}\\
\lin{\widehat{S'_1}}\cap A'\subseteq L_1 \tag{\ref{eq:allintersect0B}.1}\label{eq:allintersect1B}
\end{align}
Now we obtain using \eqref{eq:allintersect0B} that
\begin{equation}
\dim A'_i\cup\cdots\cup A'_{\ell} = r(M)-\|S'_i\|. \label{eq:allintersect0C}
\end{equation}
This implies for~$i=1$ that
\begin{equation}
\dim A'=r(M)-\|S'_1\|. \tag{\ref{eq:allintersect0C}.1}\label{eq:allintersect1C}
\end{equation}

The lemma requires establishing the existence of a vector space~$L_u$ such that
\[\dim A_i\cup L_1=\|S_i\|+\dim L_1+h\qquad\mbox{and}\qquad\lin{A_i\cup L_1}\cap\lin{A_j\cup L_1}=L_u\]
for all~$1\leq i < j \leq k$ and
such that the dimension of~$L_u$ is~$\dim L_1+h$,
where~$h$ is the distance between~$u$ and $u_1$.
Since every branch rooted at~$u$ is at capacity, we obtain using \eqref{eq:allintersect1C} that
\begin{equation}
\dim A_i\cup A'=r(M)-\left(\|S'_1\|-h-\|S_i\|\right)=\dim A'+h+\|S_i\|\label{eq:allintersect2}
\end{equation}
for every~$i=1,\ldots,k$.
Since every~$A_i$ is a subspace of~$\lin{\widehat{S'_1}}$,
we derive from \eqref{eq:allintersect1A}, \eqref{eq:allintersect1B}, and \eqref{eq:allintersect2} that
\begin{align}
\dim A_i\cup L_1 &= \dim(A_i \cup A') + \dim(\lin{A_i \cup L_1}\cap A') - \dim(A') \nonumber\\
&=\dim L_1+h+\|S_i\|.\label{eq:allintersect3}
\end{align}
Since~$(T,f)$ is a depth-decomposition, it holds that
\begin{equation}
\dim A'\cup\bigcup_{j\in J} A_j\le r(M)-\|S'_1\|+h+\sum_{j\in J}\|S_j\|=\dim A'+h+\sum_{j\in J}\|S_j\|\label{eq:allintersect4}
\end{equation}
for all~$J\subseteq\{1,\ldots,k\}$, and analogously to the proof of \eqref{eq:allintersect3},
we derive from \eqref{eq:allintersect4} that
\begin{equation}
\dim L_1\cup\bigcup_{j\in J} A_j\le\dim L_1+h+\sum_{j\in J}\|S_j\|.\label{eq:allintersect5}
\end{equation}
We now obtain using \eqref{eq:allintersect3} and \eqref{eq:allintersect5} for~$J=\{i,j\}$ that
\begin{align*}
 & \dim L_1+h+\|S_i\|+\|S_j\| \\
 & \ge \dim A_i\cup A_j\cup L_1 \\
 & = \dim A_i\cup L_1 + \dim A_j\cup L_1 - \dim \lin{A_i\cup L_1}\cap\lin{A_j\cup L_1} \\
 & = 2\dim L_1+2h+\|S_i\|+\|S_j\|-\dim \lin{A_i\cup L_1}\cap\lin{A_j\cup L_1}
\end{align*}
for all~$1\le i<j\le k$.
It follows that
\begin{equation}
\dim \lin{A_i\cup L_1}\cap\lin{A_j\cup L_1} \ge \dim L_1+h \label{eq:allintersect6}
\end{equation}
for all~$1\le i<j\le k$.
On the other hand, it holds using \eqref{eq:allintersect1D}, \eqref{eq:allintersect3} and \eqref{eq:allintersect5} applied for~$J=\{1,\ldots,k\}\setminus\{i\}$ that
\begin{align*}
 & \dim \lin{A_i\cup L_1}\cap\lin{A_j\cup L_1} \\
 & \le \dim \lin{A_i\cup L_1}\cap\lin{L_1\cup\bigcup_{j'\not=i}A_{j'}} \\
 & = \dim A_i\cup L_1 + \dim L_1\cup\bigcup_{j'\not=i}A_{j'} - \dim L_1\cup\widehat{S'_1}\\
 & \le \dim L_1+h+\|S_i\| +  \dim L_1+h+\sum_{j'\not=i}\|S_{j'}\|- \dim L_1 - \|S'_1\|\\
 & = \dim L_1+h.
\end{align*}
We conclude that equality always holds in \eqref{eq:allintersect6}.
Hence, there exists a subspace~$L_u$ of dimension~$\dim L_1+h$ such that
\begin{align*}
&\dim A_i\cup L_1=\|S_i\|+\dim L_1+h=\|S_i\|+\dim L_u\qquad\mbox{and}\\
&\lin{A_i\cup L_1}\cap\lin{A_j\cup L_1}=L_u
\end{align*}
for all~$1\leq i < j \leq k$.
\end{proof}

Using Lemma~\ref{lm:allintersect}, we prove the following.

\begin{lemma}
\label{lm:contract}
Let~$(T,f)$ be a depth-decomposition of a vector matroid~$M$,
$u_1$ a vertex of~$T$ with at least 2 children, and $u_2,\ldots,u_k$ all ancestors of~$u_1$ with at least two children (listed in the increasing distance from~$u_1$).
Assume that every branch rooted at one the vertices~$u_1,\ldots,u_k$ is at capacity, and
let~$L_1$ be the space~$L_{u_1}$ from the statement of Lemma~\ref{lm:allintersect} applied with~$U=\{u_1,\ldots,u_k\}$.
Further, let~$S_1$ be any branch rooted at~$u_1$ and $f_1$ the restriction of~$f$ to~$\widehat{S_1}$.

The pair~$(S_1,f_1)$ is a depth-decomposition of the vector matroid~$(M/L_1)\left[\widehat{S_1}\right]$ and
a branch of~$(S_1,f_1)$ is at capacity if and only if it is at capacity in~$(T,f)$.
In addition, if~$(S'_1,f'_1)$ is another depth-decomposition of the matroid~$(M/L_1)\left[\widehat{S_1}\right]$,
then~$(T',f')$ is a depth-decomposition of the matroid~$M$,
where~$T'$ is obtained from~$T$ by replacing~$S_1$ with~$S'_1$, and
the function~$f'$ is defined as~$f'(x)=f'_1(x)$ for~$x\in\widehat{S_1}$, and $f'(x)=f(x)$ otherwise.
\end{lemma}

\begin{proof}
Let~$X$ be the set of elements of~$M$, let~$A_1$ be the linear hull of~$\widehat{S_1}$, and
let~$A'_1$ be the linear hull of~$X\setminus\widehat{S_1}$.
By Lemma~\ref{lm:allintersect},
it holds that
\begin{equation}
A_1\cap A'_1\subseteq L_1\quad\mbox{and}\quad L_1\subseteq A'_1.\label{eq:contract}
\end{equation}
It follows that the matroids~$(M/L_1)\left[\widehat{S_1}\right]$ and $(M/A'_1)\left[\widehat{S_1}\right]$ are the same.
In addition, since all branches rooted at~$u_1,\ldots,u_k$ are at capacity,
it also holds that $\dim A'_1=\dim M-\|S_1\|$ by \eqref{eq:allintersect1C}.

Instead of verifying that $(S_1,f_1)$ is a depth-decomposition of~$(M/L_1)\left[\widehat{S_1}\right]$,
we verify that $(S_1,f_1)$ is a depth-decomposition of~$(M/A'_1)\left[\widehat{S_1}\right]$,
which is equivalent since the two matroids are the same.
Let~$X'$ be any subset of~$\widehat{S_1}$ and
let~$e$ be the number of edges on the paths from the~$f_1$-image of~$X'$ to the root of~$S_1$.
Since~$(T,f)$ is a depth decomposition of the matroid~$M$,
we obtain that
\[\dim X'\cup A'_1=\dim X'\cup\left(X\setminus\widehat{S_1}\right)\le\dim M-\|S_1\|+e.\]
Since the rank of the set~$X'$ in~$(M/A'_1)\left[\widehat{S_1}\right]$ is~$\dim X'\cup A'_1-\dim A'_1$ and $\dim A'_1=\dim M-\|S_1\|$,
we obtain that the rank of~$X'$ in~$(M/A'_1)\left[\widehat{S_1}\right]$ is at most~$e$ as desired.
Hence, $(S_1,f_1)$ is a depth-decomposition of~$(M/L_1)\left[\widehat{S_1}\right]$.

Let~$S$ be a branch of~$S_1$ rooted at a vertex~$u$,
let~$X'$ be the elements assigned to the leaves of the other branches rooted at~$u$, and
let~$e'$ be the number of edges contained in the other branches rooted at~$u$.
The branch~$S$ is at capacity in the depth-decomposition~$(S_1,f_1)$ of the matroid~$(M/L_1)\left[\widehat{S_1}\right]$
if and only if
the rank of~$\widehat{S_1}\setminus X'$ in the matroid~$(M/L_1)\left[\widehat{S_1}\right]$ is~$\|S_1\|-e'$.
The rank of the set~$\widehat{S_1}\setminus X'$ in the matroid~$(M/L_1)\left[\widehat{S_1}\right]$
is equal to~$\dim\left(\widehat{S_1}\setminus X'\right)\cup L_1-\dim L_1$,
which is equal to~$\dim\left(\widehat{S_1}\setminus X'\right)\cup A'_1-\dim A'_1$ by \eqref{eq:contract}.
Since~$\dim A'_1=\dim M-\|S_1\|$, we infer that
the branch~$S$ is at capacity in the depth-decomposition~$(S_1,f_1)$ of the matroid~$(M/L_1)\left[\widehat{S_1}\right]$
if and only if
$\dim\left(\widehat{S_1}\setminus X'\right)\cup A'_1$ is~$\dim M-e'$.
The latter holds if and only if~$S$ is at capacity in the depth-decomposition~$(T,f)$ of the matroid~$M$.

It remains to prove the last part of the lemma.
We proceed by induction of~$k$. The base case is~$k=1$, i.e., the case when the branch~$S_1$ is primary.
Let~$(S'_1,f'_1)$ be another depth-decomposition of the matroid~$(M/L_1)\left[\widehat{S_1}\right]$, and
let~$(T',f')$ be obtained as described in the statement of the lemma.
Denote by~$S_2,\ldots,S_{\ell}$ the remaining branches rooted at~$u_1$.
Let~$X'$ be any subset of the elements of~$M$, let~$e_i$, $i=1,\ldots,\ell$, be the number of edges on the paths from the vertices in the~$f'$-image of~$X'\cap\widehat{S_i}$ to~$u_1$, and let~$f_i$, $i\ge 2$ be the restriction of~$f$ to~$\widehat{S_i}$.
Since~$(S'_1,f'_1)$ is a depth-decomposition of the matroid~$(M/L_1)\left[\widehat{S_1}\right]$,
we obtain that
\begin{equation}
\dim \left(X'\cap\widehat{S_1}\right)\cup L_1-\dim L_1\le e_1.\label{eq:contract1}
\end{equation}
Similarly, as~$(S_i,f_i)$ is a depth-decomposition of the matroid~$(M/L_1)\left[\widehat{S_i}\right]$ by the already proven part of this lemma,
we obtain that
\begin{equation}
\dim \left(X'\cap\widehat{S_i}\right)\cup L_1-\dim L_1\le e_i.\label{eq:contract2}
\end{equation}
Using \eqref{eq:contract1} and \eqref{eq:contract2},
we infer that
\[\dim X'\le \dim X'\cup L_1\ = \dim \bigcup_{i=1}^{\ell} \left[ (X'\cap \widehat S_i)\cup L_1 \right]  \le \dim L_1+\sum_{i=1}^{\ell} e_i.\]
Since the choice of~$X'$ was arbitrary, $(T',f')$ is a depth-decomposition of~$M$.

The inductive step proceeds as follows.
Let~$k\ge 2$,
let~$S_2$ be the branch rooted at~$u_2$ in~$T$ containing~$S_1$,
let~$S'_2$ be the branch rooted at~$u_2$ in~$T'$ containing~$S'_1$, and
let~$f_2$ and $f'_2$ be the restrictions of~$f$ and $f'$, respectively, to the elements of~$\widehat{S_2}=\widehat{S'_2}$.
Finally, let~$L_2$ be the space~$L_{u_2}$ from the statement of Lemma~\ref{lm:allintersect} applied with~$U=\{u_1,\ldots,u_k\}$.
By the already proven part of the lemma, $(S_2,f_2)$ is a depth-decomposition of the matroid~$(M/L_2)\left[\widehat{S_2}\right]$.
By the base of the induction, which we have already proven,
$(S'_2,f'_2)$ is also a depth-decomposition of the matroid~$(M/L_2)\left[\widehat{S_2}\right]$.
We now apply the induction to the vertex~$u_2$ of~$T$ with replacing the branch~$S_2$ with~$S'_2$ in~$T$, and
conclude that $(T',f')$ is a depth-decomposition of the matroid~$M$.
\end{proof}

We next extend Lemma~\ref{lm:primbranch} to all branches.

\begin{lemma}
\label{lm:allbranch}
Let~$(T,f)$ be a depth-decomposition of a vector matroid~$M$, and
$S_0$ a branch of~$T$ rooted at a vertex~$u_0$ such that $S_0$ is not at capacity.
Suppose that every branch rooted at an ancestor of~$u_0$ is at capacity.
Let~$T'$ be the rooted tree obtained from~$T$ by changing the root of~$S_0$ to be the parent of~$u_0$.
Then, $(T',f)$ is a depth-decomposition of~$M$.
\end{lemma}

\begin{proof}
If~$S_0$ is primary, we apply Lemma~\ref{lm:primbranch}.
Hence, we can assume that $S_0$ is not primary.
Let~$U$ be the set of ancestors of~$u_0$ that have at least two children;
note that $U$ is non-empty since the branch~$S_0$ is not primary.
We next apply Lemma~\ref{lm:allintersect} to get subspaces~$L_u$, $u\in U$, with properties described in the statement of the lemma.
Let~$u_1$ be the vertex of~$U$ nearest to~$u_0$,
$L_1$ the vector space~$L_{u_1}$,
$S_1$ the branch rooted at~$u_1$ that contains~$u_0$, and
$f_1$ the function~$f$ restricted to the leaves of~$S_1$.
By Lemma~\ref{lm:contract},
$(S_1,f_1)$ is a depth-decomposition of the matroid~$(M/L_1)\left[\widehat{S_1}\right]$ and
the branch~$S_0$ is not at capacity in~$(M/L_1)\left[\widehat{S_1}\right]$.

Let~$S'_1$ be the rooted tree obtained from~$S_1$ by changing the root of~$S_0$ to be the parent of~$u_0$.
Since the branch~$S_0$ is primary in the depth-decomposition~$(S_1,f_1)$,
we obtain that $(S'_1,f_1)$ is a depth-decomposition of the matroid~$(M/L_1)\left[\widehat{S_1}\right]$ by Lemma~\ref{lm:primbranch}.
The fact that $(T',f)$ is a depth-decomposition of~$M$ now follows from Lemma~\ref{lm:contract}.
\end{proof}

We are now ready to present one of the main results of this section.

\begin{theorem}
\label{thm:capacity}
There exists a polynomial time algorithm that, given a vector matroid~$M$ and a depth-decomposition~$(T,f)$ of~$M$,
outputs an extended depth-decompo\-si\-tion~$(T',f',g)$ of~$M$ such that
the depth of~$T'$ is at most the depth of~$T$ and every branch of~$T'$ is at capacity.
\end{theorem}

\begin{proof}
The algorithm first modifies~$(T,f)$ to a depth-decomposition~$(T',f')$ such that every branch of~$T'$ is at capacity.
This is done iteratively as follows.
At each iteration,
the algorithm searches in the increasing order given by the distance from the root
for a vertex~$u$ with at least two children such that
a branch rooted at~$u$ is not at capacity.
The depth-decomposition is then modified by changing the root of the branch that
is not at capacity to the parent of~$u$ (note that $u$ cannot be the root of the whole tree by Lemma~\ref{lm:rootbranch}).
Lemma~\ref{lm:allbranch} implies that the new tree is again a depth-decomposition of~$M$.
This finishes the iteration and the algorithm starts a new iteration (again searching in the order given by the distance from the root).
Note that the number of leaves is preserved and at each iteration the sum of the distances of the leaves to the root of the tree decreases.
Since~$T$ has~$r(M)$ edges, it has at most~$r(M)$ leaves and each leaf is at distance at most~$r(M)$ from the root.
It follows that the algorithm stops after at most~$r(M)^2$ iterations
producing a depth-decomposition~$(T',f')$ of~$M$ such that every branch of~$T'$ is at capacity.
Since the depth of the tree is never increased by the algorithm,
the depth of~$T'$ is at most the depth of~$T$.

We next construct the function~$g$.
Let~$U$ be the set containing the root of~$T'$ and all vertices of~$(T',f')$ with at least two children, and
let~$L_u$ be the vector spaces as described in Lemma~\ref{lm:allintersect}
while setting~$L_u$ to be the space containing the zero vector only for the root~$u$ of~$T'$.
Observe that the spaces~$L_u$, $u\in U$, can be algorithmically constructed.
Indeed, if~$u\in U$ and
the space~$L_{u'}$ has already been constructed for the nearest ancestor~$u'$ of~$u$ contained in~$U$,
then~$L_u$ is the intersection of the linear hulls of~$\widehat{S_1}\cup L_{u'}$ and $\widehat{S_2}\cup L_{u'}$
where~$S_1$ and $S_2$ are any two branches rooted at~$u$.

The function~$g$ is defined for the vertices of~$T'$ in the order based on their distance from the root.
Let~$v$ be a vertex of~$T'$ and assume that $g$ has been defined for all ancestors of~$v$ (except the root).
We distinguish two cases.
The first case is that $v$ has at least two descendants that are leaves.
If~$v$ has at least two children, then let~$u$ be the vertex~$v$ itself, and
let~$u$ be the nearest descendant of~$v$ contained in~$U$ otherwise.
We set~$g(v)$ to be any vector of~$L_u$ that
is linearly independent of the~$g$-image of the vertices on the path from the parent of~$v$ to the root.
Since the dimension of~$L_u$ is the depth of~$u$, which is at least the depth of~$v$, such a vector always exists.
The other case is that $v$ and all its descendants have at most one child.
If~$v$ is a leaf, let~$v'$ be the vertex~$v$ itself. Otherwise, let~$v'$ be the only leaf descendant of~$v$.
We set~$g(v)$ to be any vector in the~$f'$-preimage of~$v'$ that
is linearly independent of the~$g$-image of the vertices on the path from the parent of~$v$ to the root.
Since the branch~$S$ containing~$v$ that is rooted at the nearest ancestor~$u$ of~$v$ contained in~$U$
is a depth-decomposition of the matroid~$(M/L_u)\left[\widehat{S}\right]$ (see Lemma~\ref{lm:contract}), such a vector always exists.
Note that the function~$g$ can be algorithmically constructed.

We now verify that $(T',f',g)$ is an extended depth-decomposition of the matroid~$M$.
Observe that for every vertex~$u\in U$, $L_u$ contains all spaces~$L_{u'}$ for ancestors~$u'$ of~$u$ contained in~$U$.
Hence, the~$g$-image of the vertices on the path from~$u$ to the root form a basis of~$L_u$ for every~$u\in U$.
In particular, $K_u=L_u$ for every~$u\in U$.
Let~$v$ be a leaf of~$T'$ and $u$ its nearest ancestor contained in~$U$.
By Lemma~\ref{lm:contract}, the branch~$S$ rooted at~$u$ containing~$v$ together with the restriction of~$f'$ to~$\widehat{S}$
is a depth-decomposition of the matroid~$(M/L_u)\left[\widehat{S}\right]$.
Note that the branch~$S$ is actually a path rooted at~$u$.
This implies that the~$g$-image of the vertices on the path from~$v$ to the child of~$u$ contained in~$S$
form a basis of the linear hull of the~$f'$-preimage of~$v$ quotiened by~$L_u$.
Since the~$g$-image of the vertices on the path from~$u$ to the root form a basis of~$K_u=L_u$,
we conclude that every vector in the~$f'$-preimage of~$v$ is contained
in the linear hull of the~$g$-image of the vertices on the path from~$v$ to the root of~$T'$.
\end{proof}

We obtain the following two statements as corollaries of Theorem~\ref{thm:capacity}.

\begin{corollary}
\label{cor:capacity}
Every vector matroid~$M$ has a depth-decomposition~$(T,f)$ with depth~$\bd(M)$ such that
every branch of~$T$ is at capacity.
\end{corollary}

\begin{corollary}
\label{cor:extdecomp}
If~$(T,f)$ is a depth-decomposition of a vector matroid~$M$,
then there exists~$g$ such that $(T,f,g)$ is an extended depth-decomposition of~$M$.
\end{corollary}

\begin{proof}
Let~$(T',f',g)$ be the extended depth-decomposition of~$M$ constructed in Theorem~\ref{thm:capacity}.
Since~$T'$ was obtained from~$T$ by rerooting some of the branches,
the vertices of~$T$ and $T'$ are in one-to-one correspondence.
In particular, 
the roots of~$T$ and $T'$ are the same vertex,
the functions~$f$ and $f'$ are identical, and
$g$ is a well-defined function from the non-root vertices of~$T$.
Further, notice that any vertex on a given root-to-leaf path in~$T'$
is also on the path from the root to the corresponding leaf in~$T$.
Since~$(T',f',g)$ is an extended depth decomposition,
any element~$x$ of~$M$ is contained in the linear hull of the~$g$-image of the vertices on the path in~$T'$ from~$f(x)$ to the root, and
thus~$(T,f,g)$ is an extended depth decomposition as well.
\end{proof}

We conclude this section with a theorem that asserts that every vector matroid
has a depth-decomposition of minimum depth such that every branch is both at capacity and solid.
Before we can state and prove the theorem, we need three auxiliary lemmas.

\begin{lemma}
\label{lm:disconnected}
Let~$M$ be a vector matroid with no loops and $M_1,\ldots,M_k$ be its components.
For each~$i$, suppose~$(T_i,f_i,g_i)$ is an extended depth-decomposition of~$M_i$.
Let~$T$ be the rooted tree obtained from the trees~$T_1,\ldots,T_k$ by identifying their roots,
let~$f$ be the mapping from the elements of~$M$ to the leaves of~$T$ such that
$f(x)=f_i(x)$ if~$x$ belongs to~$M_i$, and
let~$g$ be the mapping such that $g(v)=g_i(v)$ if~$v$ is a non-root vertex of~$T_i$.
The triple~$(T,f,g)$ is an extended depth-decomposition of~$M$.
\end{lemma}

\begin{proof}
Let~$X_1,\ldots,X_k$ be the elements of~$M$ contained in~$M_1,\ldots,M_k$, respectively.
Since~$M_1,\ldots,M_k$ are components of~$M$,
the rank of~$M$ is the sum of the ranks of~$M_1,\ldots,M_k$.
In particular, the number of edges of~$T$ is the rank of~$M$.
Since every vertex~$x\in X_i$ can be expressed as
a linear combination of the~$g_i$-image of the vertices on the path from~$f_i(x)$ to the root of~$T_i$,
it is also a linear combination of the~$g$-image of the vertices on the path from~$f(x)$ to the root of~$T$.
In particular,
the linear hull of the~$g$-image of all non-root vertices of~$T$ is equal to the linear hull of~$X_1\cup\cdots\cup X_k$.
This implies that the vectors in the~$g$-image of all non-root vertices of~$T$ are linearly independent.
Let~$X'$ be any set of vertices of~$M$.
Since any vector in~$X'$ can be expressed as a linear combination of the~$g$-image of
the non-root vertices on the paths from~$f(X')$ to the root,
the dimension of the linear hull of~$X'$ is at most the number of such vertices,
which is equal to the number of edges on the paths.
It follows that $(T,f,g)$ is an extended depth-decomposition of~$M$.
\end{proof}

\begin{lemma}
\label{lm:connected}
Let~$(T,f,g)$ be an extended depth-decomposition of a vector matroid~$M$, and
let~$u$ be a vertex with at least two children.
If~$S$ is a branch rooted at~$u$, then~$\widehat{S}$ is a union of some components and loops of~$M/K_u$.
\end{lemma}

\begin{proof}
Let~$X$ be all the vectors of~$M$, $A$ their linear hull, and $A_S$ the linear hull of the~$g$-images of the non-root vertices of~$S$.
Since the vectors of~$\image(g)$ form a basis of the vector space~$A$ and
every element~$x$ of the matroid~$M$ is a linear combination of the vectors in the~$g$-image of the vertices on the path from~$f(x)$ to the root,
$A_S$ is a subset of the linear hull of~$\widehat{S}$ and
the linear hull of~$\widehat{S}$ is a subset of the linear hull of~$A_S\cup K_u$.
Since the dimension of~$A_S$ is~$\|S\|$, we obtain that
\[\dim \widehat{S}\cup K_u-\dim K_u=\dim A_S\cup K_u-\dim K_u=\|S\|.\]
Along the same lines, we obtain that
\[\dim \left(X\setminus\widehat{S}\right)\cup K_u-\dim K_u=\dim X-\dim K_u-\|S\|.\]
Since the rank of~$M/K_u$ is~$\dim X-\dim K_u$,
the rank of~$\widehat{S}$ in~$M/K_u$ is~$\|S\|$, and
the rank of~$X\setminus\widehat{S}$ in~$M/K_i$ is~$\dim X-\dim K_u-\|S\|$.
Hence, the set~$\widehat{S}$ is a union of components and loops of~$M/K_u$.
\end{proof}

\begin{lemma}
\label{lm:replacebranch}
Let~$(T,f,g)$ be an extended depth-decomposition of a vector matroid~$M$, and
let~$u$ be a vertex with at least two children.
Further, let~$S$ be a branch rooted at~$u$ and
$(T',f',g')$ be an extended depth-decomposition of the matroid~$(M/K_u)\left[\widehat{S}\right]$.
Let~$T''$ be the rooted tree obtained by removing from~$T$ the branch~$S$ and identifying the root of~$T'$ with~$u$,
setting~$f''(x)=f'(x)$ for elements~$x\in\widehat{S}$ and $f''(x)=f(x)$ for other elements~$x$ of~$M$, and
setting~$g''(v)=g'(v)$ for non-root vertices of~$T'$ and $g''(v)=g(v)$ for other non-root vertices of~$T''$.
The triple~$(T'',f'',g'')$ is an extended depth-decomposition of~$M$.
\end{lemma}

\begin{proof}
Since the rank of the matroid~$(M/K_u)\left[\widehat{S}\right]$ is equal to~$\dim\widehat{S}\cup K_u-\dim K_u=\|S\|$,
the trees~$T'$ and $S$ have the same number of edges.
This implies that the trees~$T$ and $T''$ also have the same number of edges.
In order to establish that $(T'',f'',g'')$ is an extended depth-decomposition of~$M$,
it is now enough to verify that every element~$x$ of the matroid~$M$
is a linear combination of the~$g''$-image of the vertices on the path from~$f''(x)$ to the root.
If this is the case, then the~$g''$-image of all non-root vertices of~$T''$
form a basis of the vector space generated by the elements of~$M$, and
the rank of any subset~$X'$ of the elements of~$M$
is at most the number of non-root vertices on the paths from the~$f''$-image of~$X'$ to the root,
which is equal to the number of edges on such paths.

Let~$x$ be any element of the matroid~$M$.
If~$x\not\in\widehat{S}$,
then~$f(x)=f''(x)$ and the~$g$-image of the vertices on the path from~$f(x)=f''(x)$ to the root is the same as the~$g''$-image,
in particular, $x$ is a linear combination of the~$g''$-image of the vertices on this path.
Hence, we need to analyze the case when~$x\in\widehat{S}$.
In this case, the path from~$f''(x)$ to the root contains all vertices on the path from~$u$ to the root of~$T''$,
which implies that the linear hull of the~$g''$-image of the vertices on the path from~$f''(x)$ to the root of~$T''$ contains~$K_u$.
Since~$(T',f',g')$ is an extended depth-decomposition of the matroid~$(M/K_u)\left[\widehat{S}\right]$,
$x$ is contained in the linear hull of the union of~$K_u$ and
the~$g'$-image of the vertices on the path from~$f'(x)$ to the root of~$T'$.
Hence, $x$ is contained in the linear hull of the~$g''$-image of the vertices on the path from~$f''(x)$ to the root of~$T''$.
We conclude that $(T'',f'',g'')$ is an extended depth-decomposition of~$M$.
\end{proof}

We are now ready to prove the final theorem of this section.

\begin{theorem}
\label{thm:strongdecomp}
Every vector matroid~$M$ has an extended depth-decomposition $(T,f,g)$ of depth~$\bd(M)$ such that
every branch of~$T$ is both at capacity and solid.
\end{theorem}

\begin{proof}
We start with a depth-decomposition~$(T,f,g)$ of~$M$ with depth~$\td(M)$ and modify it iteratively as follows.
At each iteration, we first apply Theorem~\ref{thm:capacity} to obtain a depth-decomposition such that every branch is at capacity.
If every branch is solid, we stop.
If there is a branch~$S$ that is not solid, we proceed as follows.
Since~$S$ is not solid, the matroid~$(M/K_u)\left[\widehat{S}\right]$ is not connected,
where~$u$ is the root of~$S$.
Let~$M_1,\ldots,M_k$ be the components of the matroid~$(M/K_u)\left[\widehat{S}\right]$ and
let~$X_u$ be the set containing all loops of the matroid~$(M/K_u)\left[\widehat{S}\right]$.
Let~$(S_i,f_i,g_i)$ be an extended depth-decomposition of~$M_i$, $i=1,\ldots,k$, with depth~$\bd(M_i)$.
Since the branch-depth~$\bd(M_i)$ of~$M_i$ is at most the branch-depth of~$(M/K_u)\left[\widehat{S}\right]$ (as the branch-depth is a minor-monotone parameter),
the depth of each of the trees~$S_1,\ldots,S_k$ is at most the depth of~$S$.
By Lemmas~\ref{lm:disconnected} and \ref{lm:replacebranch},
it is possible to replace the branch~$S$ with the branches~$S_1,\ldots,S_k$ rooted at the root of~$S$
while assigning the elements of~$X_u$ to arbitrary leaves of the branches~$S_1,\ldots,S_k$.
Note that the depth of the new rooted tree does not exceed the depth of the original rooted tree.
In this way, we obtain a new extended depth-decomposition of~$M$, and we proceed to the next iteration.

We need to argue that the procedure described above eventually finishes.
Let~$a_i$ be the sum of the degrees of the vertices at distance~$i$ from the root.
During the iterations, the rooted tree~$T$ is modified by the algorithm presented in Theorem~\ref{thm:capacity} and
by the procedure described in the first paragraph.
The algorithm presented in Theorem~\ref{thm:capacity} selects a branch that is not at capacity and reroots it to the parent of its root.
Hence, there always exists~$i_0$ such that $a_0,\ldots,a_{i_0-1}$ are preserved and $a_{i_0}$ has increased by one.
Similarly, in the procedure described in the first paragraph, the degree of the vertex~$u$ has increased
while the degrees of all vertices with distance to the root smaller than~$u$ have not changed, in particular,
there exists~$i_0$ such that $a_0,\ldots,a_{i_0-1}$ are preserved and $a_{i_0}$ has increased.
We conclude that the vector~$(a_0,\ldots,a_{\bd(M)})$ lexicographically increases at each modification of the tree~$T$.
Since the sum of the degrees of the vertices at distance~$i$ from the root is bounded by the rank~$r$ of~$M$,
which is the total number of edges of~$T$,
there are at most~$r^{\bd(M)+1}$ vectors that can represent a sequence of sums of degrees of vertices of~$T$
at distance~$0,1,\ldots,\bd(M)$ from the root.
Hence, the procedure terminates after at most~$r^{\bd(M)+1}$ iterations.
\end{proof}

\section{Algorithm for finite fields}
\label{sec:finite}

In this section, we design a fixed parameter algorithm
for computing a depth-decomposition of a vector matroid over a fixed finite field.
To do so, we need to introduce additional notation.
Let~$(T,f,g)$ be an extended depth-decomposition of a vector matroid~$M$, and
let~$r$ be the rank of~$M$.
Let~$u_0,\ldots,u_{2r}$ be a depth-first-search transversal of the tree~$T$ (see Figure~\ref{fig:dfs} for an illustration).
For~$i\in\{0,\ldots,2r\}$,
we define~$A_i$ to be the linear hull of~$K_{u_i}$ and the~$f$-preimage of the leaves among the vertices~$u_0,\ldots,u_i$. Similarly, we define~$B_i$ to be the linear hull of~$K_{u_i}$ and the~$f$-preimage of the leaves among the vertices~$u_i,\ldots,u_{2r}$.
The sequence~$(u_i,A_i,B_i)_{i\in\{0,\ldots,2r\}}$ is called a \emph{transversal sequence} for~$(T,f,g)$.
Note that $A_i\cap B_i = K_{u_i}$ by the fact that $\image(g)$ is a basis of the linear hull of elements of~$M$.
If~$(T,f,g)$ is principal and $(T',f',g')$ is another extended depth-decomposition of~$M$,
we say that a branch~$S$ of~$T'$ is \emph{$i$-crossed}
if~$\widehat{S}$ contains the~$g$-image of a vertex on the path from~$u_i$ to the root of~$T$.

\begin{figure}
\begin{center}
\epsfbox{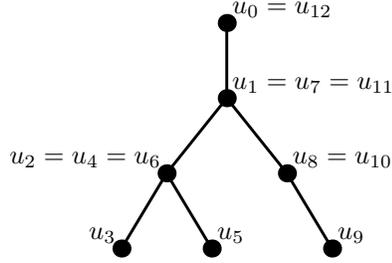}
\end{center}
\caption{An example of a depth-first-search transversal of a rooted tree.}
\label{fig:dfs}
\end{figure}

\begin{lemma}
\label{lm:subset}
Let~$M$ be a vector matroid with rank~$r$,
$(T,f,g)$ a principal extended depth-decomposition of~$M$, and
$(T',f',g')$ a solid extended depth-decomposition of~$M$.
Further, let~$(u_i,A_i,B_i)_{i\in\{0,\ldots,2r\}}$ be a transversal sequence for~$(T,f,g)$.
If~$S$ is a branch of~$(T',f',g')$ that is not~$i$-crossed,
then~$\widehat{S}$ is a subset of~$A_i\cup K_v$ or~$B_i\cup K_v$,
where~$v$ is the root of~$S$.
\end{lemma}

\begin{proof}
Let~$v$ be the root of~$S$,
$V$ the set of all vertices of~$T'$ that are not descendants of~$v$ in~$S$, and
$C$ the linear hull of~$g'(V)$.
Since~$\image(g)$ is a base of the linear hull of elements of~$M$,
the matroids~$(M/K_v)\left[\widehat{S}\right]$ and $(M/C)\left[\widehat{S}\right]$ are the same.
Since the branch~$S$ is solid in~$(T',f',g')$,
it follows that $\widehat{S}$ is the union of a component of~$M/C$ and possibly some loops corresponding to vectors contained in~$K_v\subseteq C$.

Let~$X$ be the set of elements of~$M$.
By the definition of an extended depth-decomposition,
the sets~$X\cap A_i$ and $X\cap B_i$ are unions of components and loops of the matroid~$M/K_{u_i}$.
Since~$S$ is not~$i$-crossed,
the leaf~$f'(g(v'))$ for every vertex~$v'$ on the path from~$u_i$ to the root of~$T$ is contained in~$V$.
It follows that $g(v')$ is contained in~$C$, and so~$K_{u_i}$ is a subspace of~$C$.
Hence, every component of the matroid~$M/K_{u_i}$ is a union of components and loops of the matroid~$M/C$.
In particular, each of the sets~$X\cap A_i$ and $X\cap B_i$ is a union of components and loops of the matroid~$M/C$.
Since~$\widehat{S}$ is the union of a component of~$M/C$ and possibly some vectors from~$K_v$,
it must hold that $\widehat{S}$ is a subset of the union of~$X\cap A_i$ and $K_v$ or the union of~$X\cap B_i$ and $K_v$.
The lemma follows.
\end{proof}

We will design a dynamic programming algorithm,
which will construct an optimal depth-decomposition of a vector matroid~$M$ using the information on the structure of~$M$
captured by an extended depth-decomposition of~$M$ produced by an approximation algorithm given in Theorem~\ref{thm:approx}.
The depth-decomposition will be constructed iteratively for elements of~$M$ in the order that
the leaves corresponding to them appear in the transversal sequence of the depth-decomposition produced by the approximation algorithm.
Since it would not be feasible to store all possible ``partial'' depth-decompositions,
we need a more succinct way of representing an already constructed part of a depth-decomposition,
which we now formally introduce.

Let~$T_0$ be a rooted tree;
we say that a mapping~$h$ from the non-root vertices to vectors is \emph{$k$-matchable} for some~$k\in\NN$
if there exists a surjective mapping~$g$ from~$\{1,\ldots,k\}$ to the leaves of~$T_0$ such that
for every~$j=1,\ldots,k$, the linear hull of the~$h$-image of the vertices on the path from~$g(j)$ to the root contains the~$j$-th unit vector.
A \emph{frontier} is a tuple~$(T_0,d,a,b,h)$ such that
$T_0$ is a rooted tree,
$d$, $a$, and $b$ are non-negative integers that sum to the number of edges of~$T_0$ (in particular, $T_0$ has at most~$d$ leaves), and
$h$ is a mapping from the non-root vertices of~$T_0$ to~$\FF^{d+a+b}$ such that
$\image(h)$ is a basis of~$\FF^{d+a+b}$ and $h$ is~$d$-matchable.
We will refer the middle~$a$ coordinates of~$h$-images as \emph{$A$-coordinates} and
to the last~$b$ coordinates as \emph{$B$-coordinates}.

\begin{figure}
\begin{center}
\epsfbox{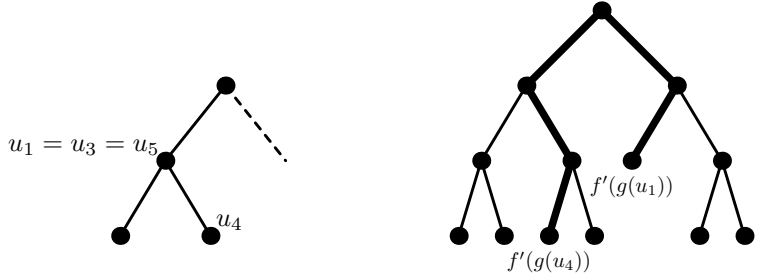}
\end{center}
\caption{An example of the construction of an~$i$-frontier.
         A part of the principal extended depth-decomposition~$(T,f,g)$ is depicted on the left and
	 the extended depth-decomposition~$(T',f',g')$ is shown on the right;
	 the subtree~$T_0$ for the~$4$-frontier is depicted in bold.}
\label{fig:frontier}
\end{figure}

Let~$(T,f,g)$ be a principal extended depth-decomposition of a vector matroid~$M$ with rank~$r$ over a field~$\FF$,
$(u_i,A_i,B_i)_{i\in\{0,\ldots,2r\}}$ a transversal sequence for~$(T,f,g)$, and
$(T',f',g')$ another extended depth-decomposition of a matroid~$M$.
The \emph{$i$-frontier} of~$(T',f',g')$ with respect to~$(T,f,g)$ and $(u_i,A_i,B_i)_{i\in\{0,\ldots,2r\}}$ is the frontier~$(T_0,d,a,b,h)$ obtained as described below;
see Figure~\ref{fig:frontier} for an illustration.
\begin{itemize}
\item The integer~$d$ is the depth of~$u_i$ in~$T$.
\item~$T_0$ is the rooted subtree of~$T'$ formed by the paths from the root of~$T'$ to the~$f'$-images of~$v^u_1,\ldots,v^u_d$,
      where~$v^u_1,\ldots,v^u_d$ are the~$g$-images of the vertices on the path from the root of~$T$ to~$u_i$ (in this order).
\item The integers~$a$ and $b$ are the smallest integers for that
      there exists an~$a$-dimensional subspace~$L_A$ of~$A_i$ and a~$b$-dimensional subspace~$L_B$ of~$B_i$ such that
      the linear hull of the~$g'$-images of the vertices of~$T_0$
      is a subspace of the linear hull of~$v^u_1,\ldots,v^u_d$, $L_A$, and $L_B$.
\item Finally, $h$ is a mapping from the non-root vertices of~$T_0$ to~$\FF^{d+a+b}$ that satisfies the following:
      Let~$v^A_1,\ldots,v^A_a$ be a basis of~$L_A$, and
      let~$v^B_1,\ldots,v^B_b$ be a basis of~$L_B$.
      The value~$h(v)$ for a non-root vertex~$v$ of~$T_0$ is equal to the coordinates of~$g'(v)$
      with respect to the (linearly independent) vectors~$v^u_1,\ldots,v^u_d,v^A_1,\ldots,v^A_a,v^B_1,\ldots,v^B_b$.
\end{itemize}
Note that we use
$i$ both as an index for the transversal sequence and as an index of the frontier in order to emphasize the link between the two indices. To see that $i$-frontiers are indeed frontiers, observe first that $K_{u_i}$ is the linear hull of~$v^u_1,\ldots,v^u_d$.
Since~$K_{u_i}=A_i\cap B_i$ and $K_{u_i}$ is contained in the linear hull of the~$g'$-images of the vertices of~$T_0$,
the subspaces~$L_A\subseteq A_i$ and $L_B\subseteq B_i$ are uniquely determined and
the dimension of the linear hull of the~$g'$-images of the vertices of~$T_0$ is equal to~$d+a+b$.
Since~$g'$ is a bijection from the non-root vertices of~$T$ to a basis of the linear hull of elements of~$M$,
the number of edges of~$T_0$ must be~$d+a+b$.
Finally, since~$T_0$ contains the~$f'$-images of~$v^u_1,\ldots,v^u_d$,
the function~$h$ is~$d$-matchable.

The following lemma justifies the definition of an \emph{$i$-frontier}.
Informally speaking, the lemma says that an~$i$-frontier splits an extended depth-decomposition into a left and a right side,
and that two depth-decompositions with the same~$i$-frontier can be glued together on it, taking one side from each decomposition;
also see Figure~\ref{fig:combine} for an illustration.
In this way, the~$i$-frontier contains all information that needs to be stored
when iteratively constructing a depth-decomposition of~$M$ in a dynamic way for the elements of contained in~$A_0,A_1,\ldots,A_{2r}$.

\begin{figure}
\begin{center}
\epsfbox{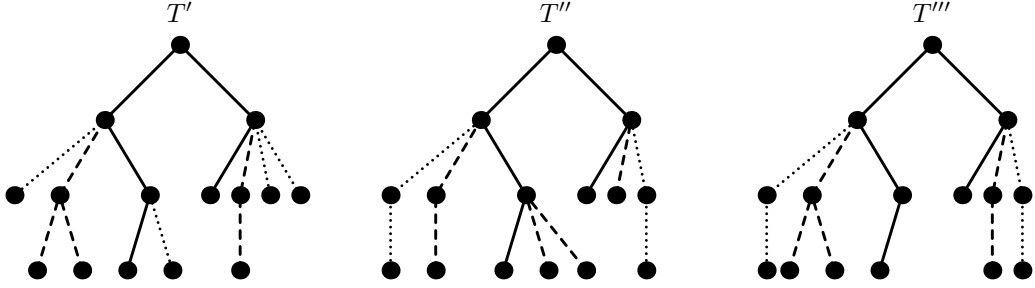}
\end{center}
\caption{An illustration of the operation described in the statement of Lemma~\ref{lm:combine}.
         The trees~$T'$, $T''$ and $T'''$ are respectively depicted on the left, middle, and right.
	 The edges of the~$i$-frontier are drawn solid,
	 the edges of branches~$S$ of~$T'$ such that $\widehat{S}\subseteq A_i\cup K_A$ and
	              branches~$S$ of~$T''$ such that $\widehat{S}\subseteq A_i\cup K_B$ are drawn dashed, and
	 the edges of branches~$S$ of~$T'$ such that $\widehat{S}\subseteq B_i\cup K_A$ and
	              branches~$S$ of~$T''$ such that $\widehat{S}\subseteq B_i\cup K_B$ are drawn dotted;
         branches of~$T'''$ are drawn in the same way as the branches of~$T'$ and $T''$ corresponding to them.
         }
\label{fig:combine}
\end{figure}

\begin{lemma}
\label{lm:combine}
Let~$(T,f,g)$ be a principal extended depth-decomposition of a vector matroid~$M$ with rank~$r$,
let~$(u_i,A_i,B_i)_{i\in\{0,\ldots,2r\}}$ be a transversal sequence for~$(T,f,g)$, and
let~$(T',f',g')$ and $(T'',f'',g'')$ be two solid extended depth-decompositions of~$M$.
Suppose that $i\in\{0,\ldots,2r\}$ is such that the~$i$-frontiers of~$(T',f',g')$ and $(T'',f'',g'')$ with respect to~$(T,f,g)$ and $(u_i,A_i,B_i)_{i\in\{0,\ldots,2r\}}$
are the same.
Let~$T_0$ be the rooted tree of the~$i$-frontier,
which we identify with the corresponding subtrees of~$T'$ and $T''$.
Further, let~$K_A$ be the linear hull of the~$g'$-image of the vertices of~$T_0$,
let~$K_B$ be the linear hull of the~$g''$-image of the vertices of~$T_0$, and
set~$C_A=K_A\cap A_i$ and $C_B=K_B\cap B_i$.

Obtain~$T'_A$ from~$T'$ by removing all branches~$S$ with~$\widehat{S}\subseteq B_i\cup K_A$ that are not~$i$-crossed,
$T''_B$ from~$T''$ by removing all branches~$S$ with~$\widehat{S}\subseteq A_i\cup K_B$ that are not~$i$-crossed, and
$T'''$ by gluing~$T'_A$ and $T''_B$ together on the vertices of~$T_0$ (note that both~$T'_A$ and $T''_B$ contain~$T_0$).
Finally, let~$f'''$ be a function from the elements of~$M$ to the leaves of~$T'''$ defined as follows.
If~$x\in A_i\setminus C_A$, then~$f'''(x)=f'(x)$.
If~$x\in B_i\setminus C_B$, then~$f'''(x)=f''(x)$.
If~$x\in C_A$,
then let~$f'''(x)$ be any leaf~$u$ of~$T_0$ such that
$x$ is contained in the linear hull of the~$g'$-image of the vertices on the path from~$u$ to the root.
Finally, if~$x\in C_B\setminus C_A$,
then let~$f'''(x)$ be any leaf~$u$ of~$T_0$ such that
$x$ is contained in the linear hull of the~$g''$-image of the vertices on the path from~$u$ to the root.
Then, $(T''',f''')$ is a depth-decomposition of~$M$.
\end{lemma}

\begin{proof}
We first verify that $f'''$ is well-defined.
Consider an element~$x$ of~$M$ and suppose that $x\in A_i\setminus C_A$.
Since the space~$C_A$ is the intersection of~$A_i$ and $K_A$,
the element~$x$ does not belong to~$K_A$ and so~$f'(x)$ is not contained in~$T_0$,
i.e., $f'(x)$ is contained in a branch of~$T'$ that is not~$i$-crossed.
Consider a maximal such branch~$S$.
Lemma~\ref{lm:subset} yields that $\widehat{S}\subseteq A_i\cup K_A$ or~$\widehat{S}\subseteq B_i\cup K_A$, and
since~$x\in A_i\setminus K_A$ and $A_i\cap B_i=K_{u_i}\subseteq K_A$, it follows that $\widehat{S}\nsubseteq B_i\cup K_A$.
Hence, the branch~$S$ is contained in~$T'_A$ and $f'''(x)$ is well-defined.
The case that $x\in B_i\setminus C_B$ is symmetric.

The next case that we analyze is that $x\in C_A$.
Since~$C_A\subseteq K_A$ and the~$g'$-image of~$T_0$ is a basis of~$K_A$, the element~$x$ can be uniquely expressed as a linear combination of the elements of the~$g'$-image of~$T_0$.
Moreover, as  $(T',f',g')$ is an extended depth-decomposition of~$M$,
the~$g'$-preimage of the basis elements with non-zero coefficients in this linear combination is contained in the path from~$f'(x)$ to the root of~$T'$.
Hence, this preimage is contained in a path from the root of~$T_0$ to one of its leaves, say~$u$.
This implies that $x\in K_u$ in~$T'$, so
$f'''(x)$ can be set to be~$u$.
The final case is that $x\in C_B\setminus C_A$.
Since~$x\in C_B$, the argument above yields that there exists a leaf~$u$ of~$T_0$ such
that $x\in K_u$ in~$T''$, and so~$f'''(x)$ can be set to be~$u$.

For completeness, we check that the sets~$A_i\setminus C_A$, $B_i\setminus C_B$, $C_A$, and $C_B\setminus C_A$ are pairwise disjoint so that $f'''(x)$ is never multiply defined. Observe that $A_i\setminus C_A=A_i\setminus K_A$ is disjoint from~$B_i$ as~$A_i\cap B_i = K_{u_i} \subseteq K_A$. Hence~$A_i\setminus C_A$ is disjoint from both~$B_i\setminus C_B$ and $C_B\setminus C_A \subseteq B_i$. Similarly, $B_i\setminus C_B=B_i\setminus K_B$ is disjoint from both~$A_i\setminus C_A$ and $C_A$. The remaining pairs are disjoint by definition, so indeed~$f'''$ is a well-defined function.

We next verify that the number of edges of~$T'''$ is the rank of~$M$,
which is
\[\dim A_i+\dim B_i-\dim A_i\cap B_i=\dim A_i+\dim B_i-\dim K_{u_i}=\dim A_i+\dim B_i-d.\]
Let~$(T_0,d,a,b,h)$ be the common~$i$-frontier of~$(T',f',g')$ and $(T'',f'',g'')$.
Observe that $\dim A_i\cap K_A=d+a$ and $\dim B_i\cap K_A=d+b$.
By Lemma~\ref{lm:subset}, 
it holds that $\widehat{S}\subseteq A_i\cup K_A$ or~$\widehat{S}\subseteq B_i\cup K_A$
for every branch~$S$ of~$T'$ that is not~$i$-crossed.
Hence, $A_i$ is contained in the linear hull of the~$g'$-image of~$T'_A$ and
the dimension of this linear hull is~$\dim A_i+b$.
Since~$\image(g')$ is a basis, the number of edges of~$T'_A$ is~$\dim A_i+b$.
A symmetric argument yields that the number of edges of~$T''_B$ is~$\dim B_i+a$.
The tree~$T_0$ has~$d+a+b$ edges since~$(T_0,d,a,b,h)$ is a frontier,
so~$T'''$ has
\[\left(\dim A_i+b\right)+\left(\dim B_i+a\right)-\left(d+a+b\right)=\dim A_i+\dim B_i-d\]
edges, as desired.

To finish the proof of the lemma, we need to show that $(T''',f''')$ is a depth-decomposition of~$M$.
To do so, we define a function~$g'''$ such that $(T''',f''',g''')$ is an extended depth-decomposition of~$M$.
Let~$u$ be a non-root vertex of~$T'''$.
If~$u$ is a non-root vertex of a branch~$S$ of~$T'$ with~$\widehat{S}\subseteq A_i\cup K_A$ that is not~$i$-crossed,
we set~$g'''(u)$ to be any nonzero vector of~$A_i$ such that $g'(u)-g'''(u)\in K_A$,
i.e., the vectors~$g'(u)$ and $g'''(u)$ are the same in the space quotioned by~$K_A$.
Similarly,
if~$u$ is a non-root vertex of a branch~$S$ of~$T''$ with~$\widehat{S}\subseteq B_i\cup K_B$ that is not~$i$-crossed,
we set~$g'''(u)$ to be any nonzero vector of~$B_i$ such that $g''(u)-g'''(u)\in K_B$.
It remains to define the mapping~$g'''$ for non-root vertices of the tree~$T_0$.
Let~$v_1^u,\ldots,v_d^u$ be the vectors assigned to the vertices on the path from~$u_i$ to the root in~$T$,
let~$v_1^A,\ldots,v_a^A$ be the basis of the space~$L_A$ as in the definition of the~$i$-frontier of~$T'$, and
let~$v_1^B,\ldots,v_b^B$ be the basis of the space~$L_B$ as in the definition of the~$i$-frontier of~$T''$.
If~$u$ is a non-root vertex of~$T_0$, we set~$g'''(u)$ to be the linear combination of the vectors
$v_1^u,\ldots,v_d^u$, $v_1^A,\ldots,v_a^A$, $v_1^B,\ldots,v_b^B$ with coefficients~$h(u)$.
Finally, let~$K'_u$, $K''_u$ and $K'''_u$ be the spaces~$K_u$ defined for trees~$T'$, $T''$ and $T'''$; we need to make this distinction here
since the spaces~$K'_u$, $K''_u$ and $K'''_u$ may differ.
Observe that for every non-root vertex~$u$ of~$T_0$,
the intersection of~$K'_u$ and $A_i$ is the same as the intersection of~$K'''_u$ and $A_i$.
The choice of~$g'''$ for the vertices of~$T'_A$ not contained in~$T_0$
implies that $K'_u\cap A_i\subseteq K'''_u\cap A_i$ for every vertex~$u$ of~$T'_A$ (indeed, even equality can be shown).
Similarly,
it holds that $K''_u\cap B_i\subseteq K'''_u\cap B_i$ for every vertex~$u$ of~$T''_B$.

It remains to show that every element~$x$ of~$M$
is a linear combination of the~$g'''$-image of the vertices on the path from~$f'''(x)$ to the root in~$T'''$.
Fix an element~$x$ of~$M$.
If~$x\in A_i\setminus C_A$, let~$u$ be the vertex~$f'''(x)=f'(x)$ and
observe that $x$ belongs to~$T'_A$ (as we have already established that $f'''(x)$ is well-defined).
Since~$K'_u\cap A_i\subseteq K'''_u\cap A_i$,
it follows that $x$ is contained in the linear hull of the~$g'''$-image of the vertices on the path from~$f'''(x)$ to the root.
If~$x\in C_A$, the vertex~$f'''(x)$ is a leaf of~$T_0$ such that
$x\in K'_u$ (as we have already established that $f'''(x)$ is well-defined).
Since~$C_A\subseteq A_i$, it follows that $x\in K'_u\cap A_i\subseteq K'''_u\cap A_i$ and
$x$ is contained in the linear hull of the~$g'''$-image of the vertices on the path from~$f'''(x)$ to the root in~$T'''$.
The cases~$x\in B_i$ and $x\in C_B\setminus C_A$ follow the same line of reasoning.
\end{proof}

To prove the main result of this section,
we will need the following auxiliary lemma.

\begin{lemma}
\label{lm:exrestrict}
Let~$(T,f,g)$ be a principal extended depth-decomposition of a vector matroid~$M$ with rank~$r$,
$(u_i,A_i,B_i)_{i\in\{0,\ldots,2r\}}$ a transversal sequence for~$(T,f,g)$, and
$(T',f',g')$ a solid extended depth-decomposition of~$M$.
The following holds for every~$i\in\{0,\ldots,2r\}$.
Let~$T_0$ be the rooted tree of the~$i$-frontier, $K$ the linear hull of the~$g'$-image of the vertices of~$T_0$,
$T_A$ the rooted tree obtained from~$T'$ by removing all branches~$S$ with~$\widehat{S}\subseteq B_i\cup K$ that are not~$i$-crossed, $g_A$ the restriction of~$g'$ to~$T_A$, and
$M_A$ the vector matroid obtained from the restriction of~$M$ to the elements of~$A_i$
by adding the~$g'$-image of the vertices of~$T_0$ (note that parallel elements may be added to~$M_A$ by this operation). There exists a function~$f_A$ from the elements of~$M_A$ to the leaves of~$T_A$ such that the triple~$(T_A,f_A,g_A)$ is an extended depth-decomposition of~$M_A$.
\end{lemma}

\begin{proof}
We define~$f_A$ as follows.
If~$x$ is contained in the restriction of~$M$ to the elements of~$A_i$ and $f'(x)$ is a leaf of~$T_A$,
we set~$f_A(x)$ to~$f'(x)$.
If~$x$ is contained in the restriction of~$M$ to the elements of~$A_i$ but~$f'(x)$ is not a leaf of~$T_A$,
we set~$f_A(x)$ to any leaf~$u$ such that $x\in K_u$.
Finally, if~$x$ is not contained in~$A_i$, then~$x=g'(v)$ for a vertex~$v$ of~$T_0$ and
we set~$f_A(x)$ to any leaf descended from~$v$.
	
We need to argue that $f_A(x)$ is well-defined for every element~$x$ of~$M_A$.
This is clear unless~$x$ is contained in the restriction of~$M$ to the elements of~$A_i$ and $f'(x)$ is not a leaf of~$T_A$.
In such a case, Lemma~\ref{lm:subset} implies that $f'(x)$ is a leaf of a branch~$S$ with~$\widehat{S}\subseteq B_i\cup K$.
It follows that the element~$x$ is contained in
\[A_i\cap \left(B_i\cup K\right)=\left(A_i\cap B_i\right)\cup\left(A_i\cap K\right)=K_{u_i}\cup\left(A_i\cap K\right)\subseteq K_{u_i}\cup K=K.\]
Since the~$g'$-image of~$T_0$ is a basis of~$K$,
the element~$x$ can be uniquely expressed as a linear combination of the elements of the~$g'$-image of~$T_0$.
Moreover, as~$(T',f',g')$ is an extended depth-decomposition of~$M$,
the~$g'$-preimage of the basis elements with non-zero coefficients in this linear combination
is contained in the path from~$f'(x)$ to the root of~$T'$.
The subpath of this path containing the preimage is also contained in~$T_0$, and
so~$T_0$ contains a path from one of its leaves, say~$u$, to its root such that the preimage is also contained in this path.
It follows that $x\in K_u$ and $f_A(x)$ can be set to~$u$.

To complete the proof,
we need to show that $(T_A,f_A,g_A)$ is an extended depth-decomposition of the matroid~$M_A$.
Observe that for every leaf~$u$ of~$T_A$,
the~$g$-image and $g_A$-image of the vertices on the path from~$u$ to the root are the same.
Hence, the space~$K_u$ is the same with respect to~$g$ and $g_A$ and it follows that $x\in K_{f_A(x)}$ for every element~$x$ of~$M_A$.
Finally,
since~$M_A$ contains all elements of~$A_i$ and a basis of~$K$ and
the~$g_A$-image of the vertices of~$T_A$ is a basis of the linear hull of~$A_i\cup K$,
the number of edges of~$T_A$ is the rank of~$M_A$.
It follows that $(T_A,f_A,g_A)$ is an extended depth-decomposition of the matroid~$M_A$.
\end{proof}

Before stating the main result of this section,
we need to establish that the number of frontiers for any fixed~$d$ is bounded.

\begin{lemma}
\label{lm:countfrontier}
For all integers~$d$ and $D$ and any finite field~$\FF$,
there exist at most~$d^{2D+1}D\lvert\FF\rvert^{(dD)^2}$ choices of a rooted tree~$T$ of depth at most~$d$, integers~$a$ and $b$, and
a mapping~$h$ from the non-root vertices of~$T$ to~$\FF^{D+a+b}$ such that $(T,D,a,b,h)$ is a frontier.
\end{lemma}

\begin{proof}
We first show that there are at most~$d^{2k-1}$ rooted trees~$T$ with~$k$ leaves and depth at most~$d$.
Such a tree can be encoded as follows: enumerate the~$k$ leaves of~$T$ in the depth-first-search order.
Let~$d_1$ be the depth of the first leaf, and
for~$i=2,\ldots,k$, let~$d_i$ be the depth of the~$i$-th leaf and
$m_i$ be the number of edges shared by the paths from the root to the~$(i-1)$-th and $i$-th leaves.
Note that the numbers~$d_1,\ldots,d_k$ and $m_2,\ldots,m_k$ determine the tree~$T$.
Since each~$d_i$ is a positive integer that is at most~$d$ and each~$m_i$ is a non-negative integer that is at most~$d-1$,
it follows that there are at most~$d^{2k-1}$ rooted trees~$T$ with~$k$ leaves and depth at most~$d$.
Summing over all the choices~$k=1,\ldots,D$,
we obtain that there are at most~$d^{2D}$ rooted trees with at most~$D$ leaves and depth at most~$d$.

Fix a tree~$T$ with at most~$D$ leaves and depth at most~$d$, and let~$m$ be the number of edges of~$T$.
Note that $m\le dD$, and $D+a+b = m$ if~$(T,D,a,b,h)$ is to be a frontier.
Hence, the integers~$a$ and $b$ can be chosen in at most~$m\le dD$ ways, and
the mapping~$h$ can be chosen in at most~$\lvert\FF\rvert^{m^2}\le \lvert\FF\rvert^{(dD)^2}$ ways (although some of these choices would not yield a frontier).
It follows that the number of choices of~$a$, $b$, and $h$, when~$T$ is fixed (which determines~$d$),
is at most~$dD\lvert\FF\rvert^{(dD)^2}$.
We conclude that for a fixed~$D$,
the number of frontiers~$(T,D,a,b,h)$ such that the depth~$T$ is at most~$d$
does not exceed~$d^{2D+1}D\lvert\FF\rvert^{(dD)^2}$.
\end{proof}

We are now ready to prove the main theorem of this section.

\begin{theorem}
\label{thm:algfin}
For the parameterization by a positive integer~$d$ and a prime power~$q$,
there exists a fixed parameter algorithm that for a vector matroid~$M$ over the~$q$-element field
either outputs that $\bd(M)$ is larger than~$d$, or
outputs an extended depth-decomposition of~$M$ with depth at most~$d$.
\end{theorem}

\begin{proof}
We first apply the algorithm from Theorem~\ref{thm:approx}.
The algorithm either outputs that the branch-depth of~$M$ is larger than~$d$ or
outputs a principal extended depth-decomposition of~$M$ with depth at most~$4^d$.
If the former applies, we stop and report that the branch-depth of~$M$ is larger than~$d$.
Otherwise, let~$(T,f,g)$ be the principal extended depth-decomposition returned by the algorithm.

Let~$r$ be the rank of the matroid~$M$,
$(u_i,A_i,B_i)_{i\in\{0,\ldots,2r\}}$ a transversal sequence for~$(T,f,g)$, and
$d_i$ the depth of~$u_i$ in~$T$.
For~$i=0,\ldots,2r$, the algorithm iteratively computes
the list of all frontiers~$(T_0,d_i,a,b,h)$ with depth at most~$d$ for which the following hold:
\begin{itemize}
\item there exists a vector matroid~$M'$ with rank~$\dim A_i+b$ such that
      the linear hull of its elements is contained in~$\lin{A_i\cup B_i}$,
\item~$M'$ contains the restriction of~$M$ to the elements of~$A_i$,
\item the matroid~$M'$ has an extended depth-decomposition~$(T_A,f_A,g_A)$ with depth at most~$d$, and
\item~$(T_0,d_i,a,b,h)$ is yielded by the procedure for creating~$i$-frontiers performed on~$(T_A,f_A,g_A)$ with respect to~$(T,f,g)$ and $(u_i,A_i,B'_i)$,
      where~$B'_i$ is intersection of the linear hull of the elements of~$M'$ and $B_i$.
      (Note that we cannot formally take~$i$-frontiers of~$(T_A,f_A,g_A)$ since~$(T,f,g)$ is not a depth-decomposition of~$M'$.)
\end{itemize}

If the branch-depth of~$M$ is at most~$d$,
then~$M$ has a solid extended depth-decomposition with depth~$\bd(M)$ by Theorem~\ref{thm:strongdecomp} and
the list is non-empty for every~$i=0,\ldots,2r$ by Lemma~\ref{lm:exrestrict}. By Lemma~\ref{lm:countfrontier}, the size of the list computed in the~$i$-th iteration does not exceed~$d^{2d_i+1}d_i\lvert\FF\rvert^{(dd_i)^2}$,
and is therefore bounded by a function of~$d$ and $\lvert\FF\rvert$ only (since~$d_i\le 4^d$ for every~$i$).
We emphasize that
$(T_0,d_i,a,b,h)$ is not required to be an~$i$-frontier of~$M'$ with respect to a principal depth-decomposition of~$M'$.

We now describe the iterations of the algorithm in detail.
For~$i=0$,
the list of frontiers contains a single element~$(R,0,0,0,h)$,
where~$R$ is the rooted tree that contains the root only and $h$ is the null function.
Hence, assume that $i>0$ and we have already computed the list for~$i-1$.
The iteration of the algorithm differs according to
whether~$u_i$ is the parent or a child of~$u_{i-1}$.

We start with the case where~$u_i$ is the parent of~$u_{i-1}$.
Then the depth of~$u_{i-1}$ is~$d_i+1$, and
the following is performed for every frontier~$(T_0,d_i+1,a,b,h)$ in the list from the previous iteration:
\begin{itemize}
\item If~$h$ is~$d_i$-matchable,
      we add~$(T_0,d_i,a+1,b,h')$ to the list for the~$i$-th iteration,
      where~$h'$ is a mapping from the non-root vertices of~$T_0$
      obtained from~$h$ by changing the~$(d_i+1)$-th coordinate into an~$A$-coordinate and
      applying an invertible linear transformation to the~$a+1$ $A$-coordinates,
      i.e., we fix such a linear transformation~$L$ and set~$h'(v)=L(h(v))$ for all vertices~$v$ of~$T_0$.
\item For every leaf~$v$ of~$T_0$ such that
      the linear hull of the~$h$-image of the vertices from~$v$ to the root contains the~$(d_i+1)$-th unit vector,
      we proceed as follows.
      Let~$T'_0$ be the tree obtained by removing the path from~$v$ to the first ancestor with at least two children, or
      to the root if there is no such ancestor.
      Let~$c$ be the number of edges on this path, and
      let~$h'$ be the restriction of~$h$ to the non-root vertices of~$T'_0$.
      If~$h'$ is~$d_i$-matchable and
      the linear hull of~$\image(h')$ restricted to the last~$b$ coordinates has dimension~$b$,
      we add~$(T'_0,d_i,a+1-c,b,h'')$ to the list for the~$i$-th iteration,
      where~$h''$ is a mapping from the non-root vertices of~$T_0'$ to~$\FF^{d_i+a+1-c+b}$
      obtained from~$h'$ by changing the~$(d_i+1)$-th coordinate into an~$A$-coordinate and
      applying any full rank linear transformation~$L:\FF^{a+1} \to \FF^{a+1-c}$ to its~$a+1$ $A$-coordinates.
\end{itemize}
We next describe the case that $u_i$ is a child of~$u_{i-1}$; note that $d_i=d_{i-1}+1$ in this case.
Let~$X_i$ be the set of~$d_i$-dimensional vectors that
formed by~$d_i$ coordinates of the elements of~$M$ contained in~$K_{u_i}$,
expressed with respect to the basis formed by the~$g$-image of the vertices on the path from the root of~$T$ to~$u_i$ (in this order).
The following is performed for every frontier~$(T_0,d_i-1,a,b,h)$ in the list from the previous iteration:
\begin{itemize}
\item For every invertible linear transformation~$L$ of the~$B$-coordinates of~$h$ and
      every leaf~$v$ of~$T_0$ such that
      the linear hull of the~$L(h)$-image of the vertices on the path from~$v$ to the root contains the~$(d_i-1+a+b)$-th unit vector,
      we proceed as follows.
      Let~$h'$ be the mapping obtained from~$L(h)$ by changing the last coordinate into a new~$d_i$-th coordinate.
      If for every~$x\in X_i$ there exists a vertex~$v'$ of~$T_0$ such that
      $x$ is contained in the linear hull of the~$h'$-image of the vertices
      on the path from~$v'$ to the root of~$T_0$ restricted to the first~$d_i$ coordinates,
      we add~$(T_0,d_i,a,b-1,h')$ to the list for the~$i$-th iteration.
\item For every rooted tree~$T'_0$ obtained from~$T_0$ by adding a new leaf~$v$ joined by a path with~$c\ge 1$ edges to~$T_0$
      in a way that the depth of~$T'_0$ is at most~$d$, we proceed as follows.
      Let~$h'$ be the mapping obtained from~$h$ by adding new~$c$ zero~$B$-coordinates and
      extending it to each vertex on the newly added path by mapping the vertex to one of the unit vector for the new coordinates in a way that
      the dimension of of the~$h'$-image is~$d_i+a+b+c-1$.
      We next consider all invertible linear transformations~$L$ of the~$b+c$ $B$-coordinates such that
      the~$L(h')$-image of the vertices on the path from~$v$ to the root contains the~$(d_i+a+b+c-1)$-th unit vector, and,
      for each such~$L$,
      we obtain the mapping~$h''$ from~$L(h')$ by turning the last coordinate into a new~$d_i$-th coordinate.
      If for every~$x\in X_i$ there exists a vertex~$v'$ of~$T'_0$ such that
      $x$ is contained in the linear hull of the~$h''$-image of the vertices
      on the path from~$v'$ to the root of~$T'_0$ restricted to the first~$d_i$ coordinates,
      we add~$(T'_0,d_i,a,b+c-1,h'')$ to the list for the~$i$-th iteration.
\end{itemize}
The inspection of the steps of the algorithm yields that
if~$(T_0,d_i,a,b,h)$ is added to the list in the~$i$-th iteration,
then~$h$ is~$d_i$-matchable and the dimension of the~$h$-image is~$d_i+a+b$.
Hence, the computed list of frontiers in the~$i$-th iteration
contains exactly the frontiers described at the beginning of the proof.

As we have already argued,
if the branch-depth of~$M$ is at most~$d$, the list of frontiers is non-empty in all iterations.
In particular, if the list becomes empty, we stop and report that the branch-depth of~$M$ exceeds~$d$.
Hence, we assume that the list is non-empty in all iterations.
If the final list is non-empty,
it contains a single element~$(T_{2r},d_{2r},a_{2r},b_{2r},h_{2r})=(R,0,0,0,h)$
where~$R$ is the rooted tree that contains the root only and $h$ is the null mapping.
We now define~$(T_i,d_i,a_i,b_i,h_i)$ for all~$i=0,\ldots,2r$
by tracing back the lists of frontiers for~$i=2r,2r-1,\ldots,1$:
if~$(T_i,d_i,a_i,b_i,h_i)$ was added to the list in the iteration~$i$,
then let~$(T_{i-1},d_{i-1},a_{i-1},b_{i-1},h_{i-1})$ be a frontier that triggered~$(T_i,d_i,a_i,b_i,h_i)$ to be added to the list.

We next use this sequence of frontiers to construct an extended depth-de\-compo\-si\-tion of depth at most~$d$.
We first define inductively for~$i=0,\ldots,2r$ linear maps~$h^M_i$ from~$\FF^{d_i+a_i}$ to the linear hull of the elements of~$M$.
The mapping~$h^M_0$ is the null mapping (note that $d_0+a_0=0$).
If~$u_i$ is a child of~$u_{i-1}$, we obtain~$h^M_i$ from~$h^M_{i-1}$ by inserting the~$d_i$-th coordinate,
mapping the~$d_i$-th unit vector to~$g(u_i)$ and extending linearly to~$\FF^{d_i+a_i}$.
If~$u_i$ is the parent of~$u_{i-1}$,
there exists a linear mapping~$L$ from~$\FF^{d_{i-1}+a_{i-1}}$ to~$\FF^{d_i+a_i}$ such that
the restriction of~$h_i$ to the first~$d_i+a_i$ coordinates is
the~$L$-transformation of the restriction of~$h_{i-1}$ to the first~$d_{i-1}+a_{i-1}$ coordinates;
we set~$h^M_i$ to be the~$L$-transformation of~$h^M_{i-1}$, i.e., $h^M_i=L(h^M_{i-1})$.

As the next step, we define inductively for~$i=0,\ldots,2r$ rooted trees~$T^M_i$ such that $T^M_i$ contains~$T_i$ as a subtree, and
mappings~$g^M_i$ from the vertices of~$T^M_i$ that are not contained in~$T_i$ to the linear hull of the elements of~$M$.
The tree~$T^M_0$ is the rooted tree that contains the root only and $g^M_0$ is the null mapping.
If~$T_i=T_{i-1}$, we set~$T^M_i=T^M_{i-1}$ and $g^M_i=g^M_{i-1}$.
Otherwise, we proceed as follows.
If~$u_i$ is a child of~$u_{i-1}$,
then~$T_i$ has been constructed from~$T_{i-1}$ by attaching a path to one of its vertices, so we obtain~$T^M_i$ by attaching a path of the same length to the corresponding vertex of~$T^M_{i-1}$ and set~$g^M_i=g^M_{i-1}$.
If~$u_i$ is the parent of~$u_{i-1}$,
we set~$T^M_i=T^M_{i-1}$, define~$g^M_i(v)=g^M_{i-1}(v)$ for vertices~$v$ of~$T^M_{i-1}$ not contained in~$T_{i-1}$ and
$g^M_i(v)=h^M_{i-1}(h_{i-1}(v))$ for vertices~$v$ of~$T_{i-1}$ not contained in~$T_i$.
Observe that $T^M_{i-1}$ is a subtree of~$T^M_i$ and $g^M_{i-1}$ is a restriction of~$g^M_i$ for all~$i=1,\ldots,2r$. Furthermore, the depth of all trees~$T^M_i$, $i=0,\ldots,2r$ is at most~$d$, and
the number of edges of~$T^M_i$ is~$\dim A_i+b_i$ for every~$i=0,\ldots,2r$.

We now establish the existence of a mapping~$f^M$ such that $(T^M_{2r},f^M,g^M_{2r})$ is an extended depth-decomposition of~$M$.
Let~$x$ be an element of~$M$ and let~$i$ be the smallest index such that $x\in K_{u_i}$.
It follows that $u_{i-1}$ is the parent of~$u_i$ and $x\in K_{u_i}\setminus K_{u_{i-1}}$.
Since~$(T_i,d_i,a_i,b_i,h_i)$ was included to the list in the iteration~$i$,
there exists a vertex~$v$ of~$T_i$ such that
the linear hull of the~$h_i$-image of the vertices on the path from~$v$ to the root of~$T_i$
contains the vector defined by the~$d_i$ coordinates of~$x$
with respect to the basis formed by the~$g$-image of the vertices on the path from the root to~$u_i$ in~$T$ (in this order).
The definition of~$g^M_i$ and the construction of the list of frontiers imply that
the the linear hull of the~$g^M_{2r}$-image of the vertices of the path from~$v$ to the root of~$T_i$ contains~$x$.
Hence, $f^M(v)$ can be chosen to be any leaf descendant of~$v$ in~$T^M_{2r}$.

We have shown that if the final list is non-empty,
the matroid~$M$ has an extended depth-decomposition of depth at most~$d$.
Since the trees~$T^M_i$ and the mappings~$h^M_i$ and $g^M_i$ can be constructed algorithmically,
the extended depth-decomposition~$(T^M_{2r},f^M,g^M_{2r})$ can also be constructed algorithmically.
Finally, observe that the number of iterations in the algorithm is twice the rank of~$M$,
the size of the list computed in each iteration is bounded by a function of~$d$ and $q$ only, and
that the number of steps needed for each element of these lists to be processed
is bounded by a function of~$d$ and $q$ times a polynomial in the size of~$M$.
Specifically, each list has at most~$d^{2^{2d+1}+1}4^dq^{(d4^d)^2}=q^{2^{O(d)}}$ elements and
the number of steps needed to process each of their elements
is bounded by a polynomial in the size of~$M$ times~$q^{(d4^d)^2}=q^{2^{O(d)}}$.
Hence the presented algorithm is a fixed parameter algorithm for parameterization by~$d$ and $q$.
\end{proof}

Theorems~\ref{thm:capacity} and \ref{thm:algfin} yield the following corollary.

\begin{corollary}
\label{cor:algfin}
For the parameterization by a positive integer~$d$ and a prime power~$q$,
there exists a fixed parameter algorithm that for a vector matroid~$M$ over the~$q$-element field
either outputs that $\bd(M)$ is larger than~$d$, or
computes~$\bd(M)$ and outputs an extended depth-decomposition with this depth such that
every branch is at capacity.
\end{corollary}

\section{Algorithm for rational matrices}
\label{sec:rational}

In this section, we adopt the algorithm presented in Section~\ref{sec:finite} to matroids over rationals.
We start with an auxiliary lemma on linear combinations appearing in extended depth-decompositions.
We remark that the bound of~$2^{2d-1}$ in Lemma~\ref{lm:express}
can be replaced with~$d\cdot 2^{d-1}$ using a slightly more careful analysis.

\begin{lemma}
\label{lm:express}
Let~$M$ be a vector matroid and $(T,f)$ a depth-decomposition of~$M$ with depth~$d$ such that every branch is at capacity.
There exists a mapping~$g$ such that $(T,f,g)$ is an extended depth-decomposition of~$M$ and
every element of~$\image(g)$ is a linear combination of at most~$2^{2d-1}$ elements of~$M$.
\end{lemma}

\begin{proof}
We show that it is possible to choose a mapping~$g$ in such a way that
the~$g$-image of a vertex at depth~$i>0$ is a linear combination of at most~$2^{d+i-1}$ elements of~$M$.
Let~$u_0,\ldots,u_k$ be the vertices on the path from the root of~$T$ to a leaf, in that order.
Let~$i_0<\cdots<i_{\ell}$ be the sequence of indices such that $i_0=0$, $i_{\ell}=k$, and $u_{i_1},\ldots,u_{i_{\ell-1}}$
are exactly the vertices among~$u_1,\ldots,u_{k-1}$ that have at least two children.
By Lemma~\ref{lm:allintersect},
each~$K_{u_{i_j}}$, $j=0,\ldots,\ell$, is the same for any function~$g$ such that $(T,f,g)$ is an extended depth-decomposition of~$M$.
Let~$L_j$ be this space, which has dimension~$i_j$ and is determined by the choice of~$T$ and $f$ only, and
note that the elements~$g(u_{i_{j-1}+1}),\ldots,g(u_{i_j})$ always form a basis of~$L_j/L_{j-1}$.
We will establish that it is possible to choose these elements in such a way that $g(u_i)$, $i=1,\ldots,k$,
is a linear combination of at most~$2^{d+i-1}$ elements of~$M$.

Suppose that we have already fixed~$g(u_1),\ldots,g(u_{i-1})$, and let~$j$ be the smallest index such that $i\le i_j$.
Note that $\dim L_j>i-1$.
If~$i_j=k$, then~$L_j/L_{j-1}$ is the linear hull of the elements in~$f^{-1}(u_k)$ quotioned by~$L_{j-1}$ and
so we choose~$g(u_i)$ to be any element of~$f^{-1}(u_k)$ that is linearly independent of~$g(u_1),\ldots,g(u_{i-1})$.
Hence, we assume that $i_j<k$ in the rest of the proof.

Let~$K$ be the linear hull of~$g(u_1),\ldots,g(u_{i-1})$, and
let~$C$ be the at most~$(1+\ldots+2^{i-1})2^{d-1}=(2^i-1)2^{d-1}$ elements of~$M$
appearing in the linear combinations used to express~$g(u_1),\ldots,g(u_{i-1})$.
Consider any two branches rooted at~$u_{i_j}$ and let~$A$ and $B$ be the~$f$-preimages of the leaves of the two branches.
By Lemma~\ref{lm:allintersect}, $L_j$ is the intersection of the linear hulls of~$A\cup L_{j-1}$ and $B\cup L_{j-1}$.
Since the linear hulls of~$A\cup L_{j-1}$ and $A\cup K$ are the same, and
the linear hulls of~$B\cup L_{j-1}$ and $B\cup K$ are the same,
$L_j$ is also the intersection of the linear hulls of~$A\cup K$ and $B\cup K$.
Since the dimension of~$L_j$ is larger than~$\dim K=i-1$,
$K$ is a proper subspace of~$L_j$ and
so the matroid~$(M/K)[A\cup B]$ contains a circuit~$X$ such that
both~$X\cap A$ and $X\cap B$ are non-empty.
Hence, $L_j \setminus K$ contains a nonzero element~$w$ that
is a linear combination of the elements of~$(X\cap A)\cup K\subseteq (X\cap A)\cup\lin{C}$ and
also a linear combination of the elements of~$(X\cap B)\cup K\subseteq (X\cap B)\cup\lin{C}$.
By Proposition~\ref{prop:circuit}, every circuit of~$M$ contains at most~$2^d$ elements, and
since every circuit of~$M/K$ is a subset of a circuit of~$M$,
we can by symmetry assume that $\lvert X\cap A\rvert\le 2^{d-1}$.
Hence, $w$ is a linear combination of at most~$2^{d+i-1}$ elements of~$M$ contained in~$(X\cap A)\cup C$ and
we can set~$g(u_i)$ to be the vector~$w$.
\end{proof}

The next lemma allows us to convert a representation of a matroid with small branch-depth over the rational numbers
to a representation of an isomorphic matroid over a finite field.

\begin{lemma}
\label{lm:largeq}
There exists an algorithm such that
\begin{itemize}
\item the input of the algorithm is an integer~$d\ge 1$ and a matroid~$M$ represented over~$\QQ$ such that
      all entries of the vectors in the representation are integers between~$-K$ and $+K$,
\item the running time of the algorithm is polynomial in the number of elements of~$M$ and $\log K$, and
\item the algorithm either outputs that the branch-depth of~$M$ is larger than~$d$ or
      computes a matroid~$M'$ represented over~$\FF_q$ for some~$q\le K^{2^{4^{d+1}}}2^{2^{2^{d+2}}}$, along with
      an isomorphism between~$M$ and $M'$.
\end{itemize}
\end{lemma}

\begin{proof}
We describe how the algorithm from the statement of the lemma proceeds.
First, the algorithm from Theorem~\ref{thm:approx} is invoked for the input matroid~$M$ and an integer~$d$.
If the algorithm outputs that the branch-depth of~$M$ is larger than~$d$, then we stop and report this.
Otherwise, we obtain a principal depth-decomposition~$(T,f,g)$ of~$M$ with depth at most~$4^d$.
For each element~$x\in M$, we compute its representation~$x'$ with respect to the basis~$\image(g)$.
Note that all entries of the vector~$x'$ are zero except for those that
correspond to the~$g$-image of the vertices on the path from~$f(x)$ to the root.
Hence, computing the entries of~$x'$ requires solving a system of at most~$4^d$ equations with integer coefficients between~$-K$ and $+K$,
which implies that the fractions appearing as the entries of~$x'$ have both numerator and denominator at most~$(K4^d)^{4^d}$ by Cramer's Rule.
We define~$\varphi(x)$ to be the integer vector obtained from~$x'$ by multiplying all its entries
by the least common multiple of the denominators of the fractions that form the entries of~$x'$.
Since the vector~$x'$ has at most~$4^d$ non-zero entries,
all entries of~$\varphi(x)$ are between~$-K'$ and $K'$ where~$K'=(K4^d)^{4^{2d}}$.

Let~$q$ be any prime larger than~$(2K'2^{4^d})^{2^{4^d}}$ and
consider the vector matroid~$M'$ over~$\FF_q$ formed by the vectors~$\image(\varphi)$.
Note that there exists such a prime~$q$ that is at most~$2\cdot (2K'2^{4^d})^{2^{4^d}}\le K^{2^{4^{d+1}}}2^{2^{2^{d+2}}}$.
Observe that $(T,f',g')$ is a depth-decomposition of~$M'$
where~$f'(\varphi(x))$ is set to~$f(x)$ and
$g'(v)$ is the unit vector whose non-zero coordinate corresponds to~$g(v)$.
Indeed, the only non-zero coordinates of the vector~$\varphi(x)$
are those corresponding to the~$g$-images of the vertices on the path from~$f(x)$ to the root.
We conclude the branch-depth of~$M'$ is at most~$4^d$.

It is well-known that if a set of integer vectors is linearly dependent over~$\QQ$,
then it is linearly dependent over any finite field;
in particular,
if~$X$ is a linearly dependent set of elements of~$M$, then~$\varphi(X)$ is linearly dependent over~$\FF_q$.
On the other hand, the choice of~$q$ yields that
if~$X$ is an independent set of at most~$2^{4^d}$ elements of~$M$,
then~$\varphi(X)$ is independent over~$\FF_q$.
Since the branch-depths of both~$M$ and $M'$ are at most~$4^d$,
neither~$M$ nor~$M'$ has a circuit with more than~$2^{4^d}$ elements by Proposition~\ref{prop:circuit}.
Hence, the matroids~$M$ and $M'$ have the same set of circuits.
It follows that the matroids~$M$ and $M'$ are isomorphic and $\varphi$ is an isomorphism between them.
\end{proof}

Using Lemmas~\ref{lm:express} and \ref{lm:largeq}, we prove the main theorem of this section.

\begin{theorem}
\label{thm:algrational}
For the parameterization by positive integers~$d$ and $K$,
there exists a fixed parameter algorithm that, for a vector matroid~$M$ over~$\QQ$ such that
the entries of all vectors in~$M$ are between~$-K$ and $+K$,
either outputs that $\bd(M)$ is larger than~$d$, or
computes~$\bd(M)$ and outputs an extended depth-decomposition~$(T,f,g)$ of~$M$ with depth~$\bd(M)$.
Moreover, the entry complexity of the vectors in~$\image(g)$ is bounded by a function of~$d$ and $K$.
\end{theorem}

\begin{proof}
Fix integers~$d$ and $K$.
We first run the algorithm from Lemma~\ref{lm:largeq} that
either outputs that the branch-depth of~$M$ is larger than~$d$ or
outputs a matroid~$M'$ with a representation over~$\FF_q$ for~$q\le K^{2^{4^{d+1}}}2^{2^{2^{d+2}}}$ that is isomorphic to~$M$ and
an isomorphism~$\varphi$ between~$M$ and $M'$.
Hence, we can use the algorithm from Corollary~\ref{cor:algfin} to decide whether the branch-depth of~$M'$ is at most~$d$, and
if so, to construct a depth-decomposition~$(T,f')$ of depth~$\bd(M')$ such that
every branch of~$T$ is at capacity.
If the branch-depth of~$M'$ exceeds~$d$, then the branch-depth of~$M$ also exceeds~$d$, so we stop and report this.
Otherwise, $(T,f)$ is a depth-decomposition of depth~$\bd(M)=\bd(M')$
where~$f$ is obtained from~$f'$ using the isomorphism between~$M$ and $M'$,
i.e., $f(x)=f'(\varphi(x))$ for every element~$x$ of~$M$.
We now use the procedure described in the proof of Lemma~\ref{lm:express} to compute a function~$g$ such that
$(T,f,g)$ is an extended depth-decomposition of~$M$.
Since computations given in the proof of Lemma~\ref{lm:express} involve solving systems of equations with at most~$2^{2d-1}$ variables,
the entry complexity of the vectors in~$\image(g)$ is bounded by~$O(d 2^{2d}\log K)$.
\end{proof}

Theorem~\ref{thm:algrational} implies Theorem~\ref{thm:alg2} as follows.

\begin{proof}[Proof of Theorem~\ref{thm:alg2}]
Let~$M$ be the matroid formed by columns of~$A$.
The algorithm from Theorem~\ref{thm:algrational}
either reports that $\bd(A)>d$ or
outputs an extended depth-decomposition~$(T,f,g)$ of~$M$ with branch-depth~$\bd(A)$ such that
the entry complexity of~$\image(g)$ is bounded by a function of~$d$ and $K$.
Let~$A'$ be the matrix from Theorem~\ref{thm:bdbasis}.
Since each element~$x$ of the matroid~$M$ is a linear combination of at most~$\bd(A)\le d$ elements from~$\image(g)$,
which are those forming the~$g$-image of the vertices on the path from~$f(x)$ to the root of~$T$,
each entry of the matrix~$A'$ can be obtained by solving a system of at most~$d$ linear equations
where the entry complexity of the coefficients and the right hand side is bounded by~$O(d 2^{2d}\log K)$.
Hence, the entry complexity of~$A'$ is bounded by~$O(d^2 2^{2d}\log K)$.
\end{proof}

\section*{Acknowledgement}

The authors would like to thank the anonymous reviewers for their comments,
which significantly improved the presentation of the paper and its results.

\bibliographystyle{bibstyle}
\bibliography{bdepth_ip}
\end{document}